\documentclass[12pt,oneside,english]{amsart}
\usepackage[T1]{fontenc}
\usepackage[latin9]{inputenc}
\usepackage{geometry}
\usepackage{babel}
\usepackage{textcomp}
\usepackage{url}
\usepackage{amstext}
\usepackage{amsthm}
\usepackage{amssymb}
\usepackage{microtype}
\usepackage[unicode=true,
 bookmarks=true,bookmarksnumbered=false,bookmarksopen=false,
 breaklinks=false,pdfborder={0 0 1},backref=false,colorlinks=false]
 {hyperref}

\makeatletter
\numberwithin{equation}{section}
\numberwithin{figure}{section}
\theoremstyle{plain}
\newtheorem{thm}{\protect\theoremname}[section]
\theoremstyle{plain}
\newtheorem{lem}[thm]{\protect\lemmaname}
\theoremstyle{plain}
\newtheorem{prop}[thm]{\protect\propositionname}
\theoremstyle{definition}
\newtheorem{defn}[thm]{\protect\definitionname}
\theoremstyle{plain}
\newtheorem{cor}[thm]{\protect\corollaryname}

\IfFileExists{lmodern.sty}{\usepackage{lmodern}}{}

\makeatother

\providecommand{\corollaryname}{Corollary}
\providecommand{\definitionname}{Definition}
\providecommand{\lemmaname}{Lemma}
\providecommand{\propositionname}{Proposition}
\providecommand{\theoremname}{Theorem}

\begin{document}
\title{Small scale equidistribution for a point scatterer on the torus}
\author{Nadav Yesha}
\address{Department of Mathematics, University of Haifa, 3498838 Haifa, Israel}
\email{nyesha@univ.haifa.ac.il}
\begin{abstract}
We study the small scale distribution of the eigenfunctions of a point
scatterer (the Laplacian perturbed by a delta potential) on two- and
three-dimensional flat tori. In two dimensions, we establish small
scale equidistribution for the ``new'' eigenfunctions holding all
the way down to the Planck scale. In three dimensions, small scale
equidistribution is established for \emph{all} of the ``new'' eigenfunctions
at certain scales.
\end{abstract}

\maketitle

\section{Introduction}

\subsection{Background}

One of the main goals in the field of Quantum Chaos is understanding
the distribution of quantum eigenstates in the semiclassical limit.
For example, the celebrated Quantum Ergodicity Theorem (\textquotedblleft Shnirelman\textquoteright s
Theorem\textquotedblright ) \cite{Shnirelman,Zelditch,ColinDeVerdiere2}
asserts that if the underlying classical dynamics of a quantum system
is ergodic, then almost all eigenstates are equidistributed in phase
space. In particular, let $\mathcal{M}$ be a smooth, compact Riemannian
manifold, and let $\left\{ \phi_{j}\right\} $ be an orthonormal basis
of $L^{2}\left(\mathcal{M},d\text{vol}\right)$ consisting of Laplace
eigenfunctions with corresponding eigenvalues $\left\{ E_{j}\right\} $
(where $d\text{vol}$ is the normalized Riemannian volume form). If
the geodesic flow on $\mathcal{M}$ is ergodic, then there exists
a density one subsequence $\left\{ \phi_{j_{k}}\right\} $, such that
for every ``nice'' $\mathcal{A}\subseteq\mathcal{M}$, we have
\begin{equation}
\int_{\mathcal{A}}\left|\phi_{j_{k}}\left(y\right)\right|^{2}d\text{vol}(y)\sim\text{vol}\left(\mathcal{A}\right)\hspace{1em}\left(k\to\infty\right).\label{eq:PositionSpaceEquidistribution}
\end{equation}
Moreover, it is expected \cite{Berry,Berry2} that generically equidistribution
should hold at \emph{smaller scales}, i.e., when $\mathcal{A}=B_{r}\left(x\right)$
is the radius $r$ geodesic ball centred at $x$, and $r$ decays
slower than the Planck scale $E_{j}^{-1/2}.$

In this paper, we study small scale equidistribution for a point scatterer,
or the Laplacian perturbed with a delta potential on the flat torus
$\mathbb{T}^{d}=\mathbb{R}^{d}/2\pi\mathbb{Z}^{d}$ ($d=2,3)$\footnote{The normalization by $2\pi$ is introduced to facilitate the notation
below. }, an important model in Quantum Chaos for studying the transition
between chaos and integrability. The underlying classical dynamics
of a toral point scatterer is integrable, since it is identical to
the geodesic flow on the torus (excluding a measure zero set of trajectories); on the other hand, numerical experiments suggest that the eigenfunctions
and spectrum of this system display chaotic features such as Gaussian-like
value distribution and level repulsion (see \cite{Seba}).

\subsection{Small scale equidistribution}

Small scale equidistribution is a very active area of research nowadays,
though most of the results are partial. On the modular surface, Luo
and Sarnak \cite{LuoSarnak} established (\ref{eq:PositionSpaceEquidistribution})
with balls $\mathcal{A}=B_{r}\left(x\right)$ of radii $r>E_{j_{k}}^{-\alpha}$
(for some small $\alpha>0)$ for a density one sequence of Hecke-Maass
forms; Young \cite{Young} showed that under the Generalized Riemann
Hypothesis, (\ref{eq:PositionSpaceEquidistribution}) holds with balls
$\mathcal{A}=B_{r}\left(x\right)$ of radii $r>E_{j}^{-1/6}$ for
$all$ such forms. Hezari and Rivière \cite{HezariRiviere} and Han
\cite{Han} established (a non-uniform version of) (\ref{eq:PositionSpaceEquidistribution})
with balls $\mathcal{A}=B_{r}\left(x\right)$ of radii $r>(\log E_{j_{k}})^{-\alpha}$
on compact negatively curved manifolds. Further results are due to
Han \cite{Han2} (small scale equidistribution for random eigenbases
on a certain class of ``symmetric'' manifolds), Han and Tacy \cite{HanTacy}
(random \emph{combinations} of Laplace eigenfunctions on compact manifolds),
Humphries \cite{Humphries} (small scale equidistribution for Hecke-Maass
forms, with balls $\mathcal{A}=B_{r}\left(x\right)$ whose \emph{centres}
are random. See also \cite{GranvilleWigman,WigmanYesha} for results
on the torus), and de Courcy-Ireland \cite{DeCourcyIreland} (discrepancy
estimates for random spherical harmonics).

An example of a manifold with small scale equidistribution holding
(almost) all the way down to Planck scale is the the two-dimensional
torus $\mathbb{T}^{2}$, as was demonstrated by Lester and Rudnick
\cite{LesterRudnick} (see also Hezari and Rivière \cite{HezariRiviere2}),
who showed that for every orthonormal basis $\left\{ \phi_{j}\right\} $
of toral Laplace eigenfunctions, there exists a density one subsequence
$\left\{ \phi_{j_{k}}\right\} $ such that
\begin{equation}
\int_{B_{r}\left(x\right)}\left|\phi_{j_{k}}\left(y\right)\right|^{2}d\text{vol}(y)\sim\text{vol}\left(B_{r}\left(x\right)\right)\hspace{1em}\left(k\to\infty\right)\label{eq:SmallScaleEquidistribution}
\end{equation}
uniformly for $r>E_{j_{k}}^{-1/2+o\left(1\right)}$ and $x\in\mathbb{T}^{2}$.
More generally, for the $d$-dimensional torus $\mathbb{T}^{d}$,
they established (\ref{eq:SmallScaleEquidistribution}) for a density
one sequence of toral Laplace eigenfunctions uniformly for $r>E_{j_{k}}^{-\frac{1}{2\left(d-1\right)}+o\left(1\right)}$
and $x\in\mathbb{T}^{d}$, and proved that the lower bound on the
radii is sharp (a refined version of Lester and Rudnick's result for
the two-dimensional torus was obtained by Granville and Wigman \cite{GranvilleWigman}
more recently).

\subsection{Toral point scatterers}

A point scatterer on the torus is formally defined as the rank one
singular perturbation
\begin{equation}
-\Delta+\alpha\left\langle \delta_{x_{0}},\cdot\right\rangle \delta_{x_{0}}\label{eq:PointScatterer}
\end{equation}
where $\Delta$ is the Laplace-Beltrami operator, $\alpha\in\mathbb{R}$
is a coupling parameter, and $\delta_{x_{0}}$ is the Dirac delta
potential at $x_{0}\in\mathbb{T}^{d}.$ Rigorously, as described in
\cite{ColinDeVerdiere}, the operator (\ref{eq:PointScatterer}) is
realized as a self-adjoint operator acting on $L^{2}\text{\ensuremath{\left(\mathbb{T}^{d}\right)} }$
via the theory of self-adjoint extensions. One begins with the Laplacian
acting on the domain of smooth functions vanishing near $x_{0}$;
there exists a one parameter family of self-adjoint extensions of
this operator denoted by $\Delta_{\phi}$ with $\phi\in(-\pi,\pi]$,
each corresponding to an operator (\ref{eq:PointScatterer}) with
a particular value of $\alpha$ ($\phi=\pi$ gives the trivial extension
$\Delta_{\pi}=\Delta$ corresponding to $\alpha=0$, the only extension
for manifolds of dimension $d\ge4$).

The spectrum of a non-trivial self-adjoint extension consists of two
types of eigenvalues:
\begin{itemize}
\item The ``old'', or unperturbed eigenvalues -- these are the nonzero
Laplace eigenvalues with multiplicities reduced by $1$. The corresponding
eigenfunctions are precisely the Laplace eigenfunctions vanishing
at $x_{0}$. 
\item A set $\Lambda=\Lambda_{\phi}$ of ``new'' or ``perturbed'' eigenvalues,
of multiplicity $1$ each, interlacing with the Laplace eigenvalues,
so that there is a unique new eigenvalue between every two Laplace
eigenvalues (its exact position depends on the choice of the self-adjoint
extension). For a new eigenvalue $\lambda\in\Lambda$, the corresponding
$L^{2}$-normalized new eigenfunction is $g_{\lambda}=G_{\lambda}/\left\Vert G_{\lambda}\right\Vert _{2}$,
where $G_{\lambda}$ is Green's function $G_{\lambda}\left(x;x_{0}\right)=\left(\Delta+\lambda\right)^{-1}\delta_{x_{0}}.$
\end{itemize}
The semiclassical limits of the new eigenfunctions of point scatterers
on flat tori have been extensively studied in recent years (for a
survey on some of the results, see \cite{Ueberschaer}). Rudnick and
Ueberschär \cite{RudnickUeberschaer} showed that for a point scatterer
on a two-dimensional torus, a density one subsequence of the new eigenfunctions
are equidistributed in configuration space, i.e., (\ref{eq:PositionSpaceEquidistribution})
holds along a density one subset $\Lambda'\subseteq\Lambda$; by a
density one subset we mean that
\[
\lim_{X\to\infty}\frac{\#\left\{ \lambda\in\Lambda':\,\lambda\le X\right\} }{\#\left\{ \lambda\in\Lambda:\,\lambda\le X\right\} }=1.
\]
In \cite{Yesha1}, it was shown that for a point scatterer on the
standard three-dimensional torus $\mathbb{T}^{3}$, equidistribution
in configuration space holds for \emph{all} of the new eigenfunctions
(and along a density one subsequence of the new eigenfunctions for
point scatterers on tori with a Diophantine aspect ratio). Recently,
we were able to establish equidistribution in configuration space
for tori with two point scatterers \cite{Yesha3}.

Equidistribution in full phase space (along a density one subsequence)
was established both on the standard two-dimensional torus $\mathbb{T}^{2}$
by Kurlberg and Ueberschär \cite{KurlbergUeberschaer}, and on the
standard three-dimensional torus $\mathbb{T}^{3}$ \cite{Yesha2}.
The quantum limits of a point scatterer on a torus with an irrational
aspect ratio (also known as the Šeba billiard \cite{Seba}) were further
studied by Kurlberg-Ueberschär \cite{KurlbergUeberschaer2}, who proved
the existence of ``scars'', i.e., localized quantum limits. The
existence of scars for arithmetic point scatterers was established
by Kurlberg and Rosenzweig \cite{KurlbergRosenzweig}.

\subsection{Statement of the main results}

We now state our main results concerning small scale equidistribution
for the new eigenfunctions $g_{\lambda}$ of toral point scatterers.
We only discuss the \emph{standard} flat two- and three-dimensional
tori, although using our arguments, analogous (albeit weaker) results
may also be obtained for any torus of the form $\mathbb{T}^{2}_{\mathcal{L}_0}=\mathbb{R}^{2}/2\pi\mathcal{L}_0$ where $ \mathcal{L}_0 = \mathbb{Z}(1/a,0)\oplus\mathbb{Z}(0,a)$ is a unimodular lattice (as in \cite{RudnickUeberschaer}), and for three-dimensional tori with Diophantine aspect ratio (as in \cite{Yesha1}). The
first principal result asserts the small scale equidistribution for
the new eigenfunctions of a point scatterer on the standard flat two-dimensional
torus $\mathbb{T}^{2}$, holding (almost) all the way down to the
Planck scale. In particular, we significantly strengthen the main
result in \cite{RudnickUeberschaer} in this case.
\begin{thm}
\label{thm:MainThm2d}Let $d=2,$ and fix $\phi\in\left(-\pi,\pi\right).$
There is a subset $\Lambda'\subseteq\Lambda_{\phi}$ of density one,
such that for every $\epsilon>0$,
\[
\sup_{\substack{r>\lambda^{-1/2+\epsilon}\\
x\in\mathbb{T}^{2}
}
}\left|\frac{\int_{B_{r}\left(x\right)}\left|g_{\lambda}\left(y\right)\right|^{2}dy}{\pi r^{2}}-1\right|\to0
\]
as $\lambda\to\infty$ along $\Lambda'$.
\end{thm}

Next, we establish the small scale equidistribution for the new eigenfunctions
of a point scatterer on the standard flat three-dimensional torus
$\mathbb{T}^{3}$. For balls with radii $r>\lambda^{-1/12+o\left(1\right)},$
our statement will hold for $all$ new eigenfunctions, improving upon
the principal result in \cite{Yesha1}.
\begin{thm}
\label{thm:MainThm3d}Let $d=3$, and fix $\phi\in\left(-\pi,\pi\right).$
For every $\epsilon>0$,
\[
\sup_{\substack{r>\lambda^{-1/12+\epsilon}\\
x\in\mathbb{T}^{3}
}
}\left|\frac{\int_{B_{r}\left(x\right)}\left|g_{\lambda}\left(y\right)\right|^{2}dy}{\frac{4}{3}\pi r^{3}}-1\right|\to0
\]
as $\lambda\to\infty$ along $\Lambda_{\phi}$.
\end{thm}

If one is willing to exclude a thin set of energy levels, the exponent
of the radii in the three-dimensional case can be improved from $-1/12+\epsilon$
to $-1/6+\epsilon$. 

\begin{thm}
\label{thm:Thm3dAE}Let $d=3,$ and fix $\phi\in\left(-\pi,\pi\right).$
There is a subset $\Lambda'\subseteq\Lambda_{\phi}$ of density one,
such that for every $\epsilon>0$,
\[
\sup_{\substack{r>\lambda^{-1/6+\epsilon}\\
x\in\mathbb{T}^{3}
}
}\left|\frac{\int_{B_{r}\left(x\right)}\left|g_{\lambda}\left(y\right)\right|^{2}dy}{\frac{4}{3}\pi r^{3}}-1\right|\to0
\]
as $\lambda\to\infty$ along $\Lambda'$.
\end{thm}

It would be interesting to determine whether equidistribution holds down to the Planck scale in dimension three, or in a more restricted range imposed by the arithmetic setting (compare with the case $ d>2 $ in \cite{LesterRudnick}).

In proving theorems \ref{thm:MainThm2d}, \ref{thm:MainThm3d}, and
\ref{thm:Thm3dAE}, we exploit the interlacing property of the new
eigenvalues in an essential way; however, we do not use any further
information regarding the exact position of the new eigenvalues. Thus,
all of our results can be easily formulated for $\lambda$-dependent
sequence of $\phi$. This is of significance, since the in the physics
literature one often considers self-adjoint extensions with $\phi$
varying with $\lambda$ (``strong coupling'', see \cite{Shigehara,Ueberschaer}).

\subsection*{Acknowledgements}

The author would like to express his gratitude to Z. Rudnick and I.
Wigman for useful discussions and comments. The research leading to
these results was partially supported by the European Research Council
under the European Union\textquoteright s Seventh Framework Programme
(FP7/2007-2013), ERC grant agreement n° 335141.

\section{Point scatterers on the torus}

\subsection{The spectrum of the toral Laplacian}

Let $\mathbb{T}^{d}=\mathbb{R}^{d}/2\pi\mathbb{Z}^{d}$ ($d=2,3$)
be the standard flat $d$-dimensional torus, and let $\Delta$ be
the associated Laplace-Beltrami operator. The spectrum of $-\Delta$
is the collection of all numbers that can be expressed as a sum of
$d$ squares, denoted by 
\[
\mathcal{N}_{d}=\left\{ 0=n_{1}<n_{2}<\dots\right\} .
\]
Recall that by Landau's Theorem,
\begin{equation}
\#\left\{ n\in\mathcal{N}_{2}:\,n\le X\right\} \sim K\frac{X}{\sqrt{\log X}},\label{eq:LandauTheorem}
\end{equation}
where $K=2^{-1/2}\prod_{p\equiv3\,\left(4\right)}\left(1-p^{-2}\right)^{-1/2}=0.764\dots$
is the Landau-Ramanujan constant. In three-dimensions, $n\in\mathcal{N}_{3}$
if and only if in the representation $n=4^{a}n_{1}$ with $4\nmid n_{1},$
the number $n_{1}$ satisfies $n_{1}\not\equiv7\,\left(8\right).$
Moreover, as $X\to\infty$,
\[
\#\left\{ n\in\mathcal{N}_{3}:\,n\le X\right\} \sim\frac{5}{6}X.
\]
Let $r_{d}\left(n\right)$ be the number of representations of $n\in\mathcal{N}_{d}$
as a sum of $d$ squares. For $d=2$, it is well-known that 
\begin{equation}
r_{2}\left(n\right)=O_{\eta}\left(n^{\eta}\right)\label{eq:r2nBound}
\end{equation}
 (in fact, (\ref{eq:LandauTheorem}) implies that on average, $r_{2}\left(n\right)$
is of order of magnitude $\sqrt{\log n}$). For $d=3$, Siegel's Theorem
\cite{Siegel} implies that for $n=4^{a}n_{1}$ with $4\nmid n_{1},$
we have 
\begin{equation}
n_{1}^{\frac{1}{2}-\eta}\ll_{\eta}r_{3}\left(n\right)=r_{3}\left(n_{1}\right)\ll_{\eta}n_{1}^{\frac{1}{2}+\eta}.\label{eq:SiegelBound}
\end{equation}

\subsection{Toral point scatterers}

Let $x_{0}\in\mathbb{T}^{d}.$ A point scatterer, formally defined
in (\ref{eq:PointScatterer}), can be rigorously realized via the
following procedure: denote by $D_{0}=C_{c}^{\infty}\left(\mathbb{T}^{d}\setminus\left\{ x_{0}\right\} \right)$
the space of smooth functions supported away from the point $x_{0}$,
and denote by $-\Delta_{0}=-\Delta_{\restriction D_{0}}$ the Laplacian
restricted to this domain. This is a symmetric operator with deficiency
indices $\left(1,1\right)$, hence there is a one parameter family
of self-adjoint extensions, which we denote by $-\Delta_{x_{0},\phi}$,
$\phi\in(-\pi,\pi]$. For $\phi\ne\pi$ (the extension $\phi=\pi$
retrieves the standard Laplacian), the spectrum of $-\Delta_{x_{0},\phi}$
consists of two types of eigenvalues:
\begin{itemize}
\item The ``old'', nonzero Laplace eigenvalues $0\ne n\in\mathcal{N}_{d}$,
which correspond to the Laplace eigenfunctions vanishing at $x_{0}$.
\item A set $\Lambda=\Lambda_{\phi}$ of ``new'' or ``perturbed'' eigenvalues,
which are the solutions to the equation
\[
\sum_{\xi\in\mathbb{Z}^{d}}\left(\frac{1}{\left|\xi\right|^{2}-\lambda}-\frac{\left|\xi\right|^{2}}{\left|\xi\right|^{4}+1}\right)=c_{0}\tan\frac{\phi}{2}
\]
where 
\[
c_{0}=\sum_{\xi\in\mathbb{Z}^{d}}\frac{1}{\left|\xi\right|^{4}+1}.
\]
\end{itemize}
Thus, the elements of $\Lambda$ interlace with the elements of $\mathcal{N}_{d}$,
i.e., between every two elements $n_{k},n_{k+1}\in\mathcal{N}_{d}$
there is a unique element of $\Lambda$, so we label $\Lambda=\left\{ \lambda_{0},\lambda_{1},\lambda_{2},\dots\right\} $,
where
\begin{equation}
\lambda_{0}<n_{1}<\lambda_{1}<n_{2}<\lambda_{2}<n_{3}<\lambda_{3}<\dots.\label{eq:Interlacing}
\end{equation}
As mentioned above, all of our results will still hold for an arbitrary
set $\Lambda$ whose elements interlace with the elements of $\mathcal{N}_{d}$.
The corresponding eigenfunctions, which are our main objects of study,
are multiples of Green's functions $\left(\Delta+\lambda\right)^{-1}\delta_{x_{0}}$,
admitting the $L^{2}$ expansion
\[
G_{\lambda}\left(x;x_{0}\right):=-\left(2\pi\right)^{d/2}\left(\Delta+\lambda\right)^{-1}\delta_{x_{0}}=\sum_{\xi\in\mathbb{Z}^{d}}\frac{e^{i\left\langle x-x_{0},\xi\right\rangle }}{\left|\xi\right|^{2}-\lambda}.
\]
Let 
\[
g_{\lambda}\left(x\right)=g_{\lambda}\left(x;x_{0}\right)=\frac{G_{\lambda}\left(x;x_{0}\right)}{\left\Vert G_{\lambda}\right\Vert _{2}}
\]
be the $L^{2}$-normalized eigenfunctions.

We will also work with a truncated version of the Green's functions
$G_{\lambda}$. Let $L=\lambda^{\delta}$, $0<\delta<1$, and define
the truncated Green's function
\[
G_{\lambda,L}\left(x;x_{0}\right)=\sum_{\substack{\xi\in\mathbb{Z}^{d}\\
\left|\left|\xi\right|^{2}-\lambda\right|<L
}
}\frac{e^{i\left\langle x-x_{0},\xi\right\rangle }}{\left|\xi\right|^{2}-\lambda}
\]
and the $L^{2}$-normalized truncated Green's function
\begin{equation}
g_{\lambda,L}\left(x\right)=g_{\lambda,L}\left(x;x_{0}\right)=\frac{G_{\lambda,L}\left(x;x_{0}\right)}{\left\Vert G_{\lambda,L}\right\Vert _{2}}.\label{eq:TruncatedGreenFunction}
\end{equation}

\section{Proof of Theorem \ref{thm:MainThm2d} \texorpdfstring{($d=2$)}{(d=2)}}

\subsection{Preliminary lemmas}

We now take $d=2$, and recall the following lemma, proved in \cite{RudnickUeberschaer}
for two-dimensional flat tori with more general aspect ratios, which
shows that along a density one subsequence, the distance between consecutive
elements in $\mathcal{N}_{2}$ is small, so that almost all $\lambda\in\Lambda$
are close to an element of $\mathcal{N}_{2}$, implying a lower bound
on the $L^{2}$-norm of $G_{\lambda}$. Recall the labeling (\ref{eq:Interlacing}),
so that given $\lambda\in\Lambda$, we can denote $\lambda=\lambda_{k}$,
where
\[
\dots<\lambda_{k-1}<n_{k}<\lambda_{k}<n_{k+1}<\lambda_{k+1}<\dots.
\]

\begin{lem}[{\cite[Lemma 2.1, Lemma 4.1]{RudnickUeberschaer}}]
\label{lem:LowerBound}~
\begin{enumerate}
\item Along a density one subsequence $\left\{ n_{k_{j}}\right\} \subseteq$$\mathcal{N}_{2}$,
the spacings of $\mathcal{N}_{2}$ satisfy 
\[
n_{k_{j}+1}-n_{k_{j}}\ll_{\eta}n_{k_{j}}^{\eta}
\] for every $ \eta>0 $.

\item There is a subset $\Lambda'$ of density one in $\Lambda$, such
that for all $\lambda_{k_{j}}\in\Lambda'$ and all $ \eta>0 $, we have 
\[
n_{k_{j}+1}-n_{k_{j}}\ll_{\eta}\lambda_{k_{j}}^{\eta}.
\]

\item For all $\lambda\in\Lambda'$ and all $ \eta>0 $, we have $\left\Vert G_{\lambda}\right\Vert _{2}\gg_{\eta}\lambda^{-\eta}$.

\end{enumerate}

\end{lem}

\begin{lem}
\label{lem:TruncationLemma}There exists a density one subset $\Lambda'\subseteq\Lambda$,
such that for every $\eta>0$, $0<\delta<1$, we have
\begin{equation}
\left\Vert g_{\lambda}-g_{\lambda,L}\right\Vert _{2}^{2}=O_{\eta}\left(\frac{\lambda^{\eta}}{L}\right)\label{eq:TruncationLimit}
\end{equation}
as $\lambda\to\infty$ along $\lambda\in\Lambda'$, where $L=\lambda^{\delta}.$
\end{lem}

\begin{proof}
We have
\begin{align*}
\left\Vert g_{\lambda}-g_{\lambda,L}\right\Vert _{2} & =\frac{1}{\left\Vert G_{\lambda,L}\right\Vert _{2}\left\Vert G_{\lambda}\right\Vert _{2}}\left\Vert \left\Vert G_{\lambda,L}\right\Vert _{2}G_{\lambda}-\left\Vert G_{\lambda}\right\Vert _{2}G_{\lambda,L}\right\Vert _{2}\\
 & \le\frac{\left\Vert G_{\lambda}-G_{\lambda,L}\right\Vert _{2}}{\left\Vert G_{\lambda}\right\Vert _{2}}+\frac{1}{\left\Vert G_{\lambda}\right\Vert _{2}}\left|\left\Vert G_{\lambda,L}\right\Vert _{2}-\left\Vert G_{\lambda}\right\Vert _{2}\right|\\
 & \le2\frac{\left\Vert G_{\lambda}-G_{\lambda,L}\right\Vert _{2}}{\left\Vert G_{\lambda}\right\Vert _{2}}.
\end{align*}
Now,
\begin{align}
\left\Vert G_{\lambda}-G_{\lambda,L}\right\Vert _{2}^{2} & =\sum_{\left|\left|\xi\right|^{2}-\lambda\right|\ge L}\frac{1}{\left(\left|\xi\right|^{2}-\lambda\right)^{2}}=\sum_{\begin{subarray}{c}
n\ge0\\
\left|n-\lambda\right|\ge L
\end{subarray}}\frac{r_{2}\left(n\right)}{\left(n-\lambda\right)^{2}}\label{eq:TruncationUpperBnd}\\
 & \ll_{\eta}\lambda^{\eta}\sum_{\left|n-\lambda\right|\ge L}\frac{1}{\left(n-\lambda\right)^{2}}\ll\lambda^{\eta}\int_{\left|x-\lambda\right|\ge L/2}\frac{1}{\left(x-\lambda\right)^{2}}\,dx\ll\frac{\lambda^{\eta}}{L}.\nonumber 
\end{align}
Thus, combining (\ref{eq:TruncationUpperBnd}) with part (3) of Lemma
\ref{lem:LowerBound}, the bound (\ref{eq:TruncationLimit}) follows
for the density one subset $\Lambda'$ established in part (2) of
Lemma \ref{lem:LowerBound}.
\end{proof}
We use the truncated Green's functions (\ref{eq:TruncatedGreenFunction})
to approximate the $L^{2}$-mass of $g_{\lambda}$ restricted to the
ball $B_{x}\left(r\right)$.
\begin{lem}
\label{lem:ApproximationLemma}There exists a density one subset $\Lambda'\subseteq\Lambda$,
such that for every $\eta>0$, $0<\delta<1$, we have 
\[
\sup_{\substack{x\in\mathbb{T}^{2}\\
r\in\mathbb{R}
}
}\left|\int_{B_{x}\left(r\right)}\left|g_{\lambda}\left(y\right)\right|^{2}\,dy-\int_{B_{x}\left(r\right)}\left|g_{\lambda,L}\left(y\right)\right|^{2}\,dy\right|=O_{\eta}\left(\frac{\lambda^{\eta}}{L^{1/2}}\right)
\]
as $\lambda\to\infty$ along $\lambda\in\Lambda'$, where $L=\lambda^{\delta}.$
\end{lem}

\begin{proof}
Let $x\in\mathbb{T}^{2}$, $r\in\mathbb{R}.$ Let $\mathbf{1}_{B_{x}\left(r\right)}$
be the indicator function of $B_{x}\left(r\right).$ Then by the Cauchy-Schwarz
inequality and by (\ref{eq:TruncationLimit}),
\begin{align*}
 & \left|\int_{B_{x}\left(r\right)}\left|g_{\lambda}\left(y\right)\right|^{2}\,dy-\int_{B_{x}\left(r\right)}\left|g_{\lambda,L}\left(y\right)\right|^{2}\,dy\right|\\
 & =\left|\left\langle \mathbf{1}_{B_{x}\left(r\right)}g_{\lambda},g_{\lambda}\right\rangle -\left\langle \mathbf{1}_{B_{x}\left(r\right)}g_{\lambda,L},g_{\lambda,L}\right\rangle \right|\\
 & \le\left|\left\langle \mathbf{1}_{B_{x}\left(r\right)}\left(g_{\lambda}-g_{\lambda,L}\right),g_{\lambda}\right\rangle \right|+\left|\left\langle \mathbf{1}_{B_{x}\left(r\right)}g_{\lambda,L},g_{\lambda}-g_{\lambda,L}\right\rangle \right|\\
 & \le\left\Vert \mathbf{1}_{B_{x}\left(r\right)}\left(g_{\lambda}-g_{\lambda,L}\right)\right\Vert _{2}+\left\Vert \mathbf{1}_{B_{x}\left(r\right)}g_{\lambda,L}\right\Vert _{2}\left\Vert g_{\lambda}-g_{\lambda,L}\right\Vert _{2}\\
 & \le2\left\Vert g_{\lambda}-g_{\lambda,L}\right\Vert _{2}\ll_{\eta}\frac{\lambda^{\eta}}{L^{1/2}},
\end{align*}
for $\lambda\in\Lambda'$, the density one subset established in part
(2) of Lemma \ref{lem:LowerBound}. The statement of the lemma follows.
\end{proof}
In light of Lemma \ref{lem:ApproximationLemma}, Theorem \ref{thm:MainThm2d}
will immediately follow from the following main proposition for the
truncated Green's functions $g_{\lambda,L}$, which will be proved
in the following subsections.
\begin{prop}
\label{prop:MainProp}There exists a density one subset $\Lambda'\subseteq\Lambda$,
such that for every $\epsilon>0,$
\[
\sup_{\substack{x\in\mathbb{T}^{2}\\
r>\lambda^{-1/2+\epsilon}
}
}\left|\frac{1}{\pi r^{2}}\int_{B_{x}\left(r\right)}\left|g_{\lambda,L}\left(y\right)\right|^{2}\,dy-1\right|\to0
\]
as $\lambda\to\infty$ along $\Lambda'$, where $L=\lambda^{\delta},$
$0<\delta<\epsilon/6$.
\end{prop}

\begin{proof}[Proof of Theorem \ref{thm:MainThm2d} assuming Proposition \ref{prop:MainProp}]
 Take $\Lambda'\subseteq\Lambda$ to be the intersection of the density
one subsets whose existence has been established in Lemma \ref{lem:ApproximationLemma}
and Proposition \ref{prop:MainProp}. By Lemma \ref{lem:ApproximationLemma},
for every $x\in\mathbb{T}^{2},\,r>\lambda^{-1/2+\epsilon},$ we have
\begin{align*}
\left|\frac{1}{\pi r^{2}}\int_{B_{x}\left(r\right)}\left|g_{\lambda}\left(y\right)\right|^{2}\,dy-1\right| & \le\left|\frac{1}{\pi r^{2}}\int_{B_{x}\left(r\right)}\left|g_{\lambda,L}\left(y\right)\right|^{2}\,dy-1\right|+O_{\eta}\left(\lambda^{\eta-\delta/2}\right)
\end{align*}
as $\lambda\to\infty$ along $\Lambda'$. Hence,
\begin{align*}
\sup_{\substack{x\in\mathbb{T}^{2}\\
r>\lambda^{-1/2+\epsilon}
}
}\left|\frac{1}{\pi r^{2}}\int_{B_{x}\left(r\right)}\left|g_{\lambda}\left(y\right)\right|^{2}\,dy-1\right| & \le\sup_{\substack{x\in\mathbb{T}^{2}\\
r>\lambda^{-1/2+\epsilon}
}
}\left|\frac{1}{\pi r^{2}}\int_{B_{x}\left(r\right)}\left|g_{\lambda,L}\left(y\right)\right|^{2}\,dy-1\right|\\
 & +O_{\eta}\left(\lambda^{\eta-\delta/2}\right)
\end{align*}
as $\lambda\to\infty$ along $\Lambda'$. Theorem \ref{thm:MainThm2d}
now follows from Proposition \ref{prop:MainProp}, choosing $0<\delta<\epsilon/6$
and $\eta<\delta/2$.
\end{proof}

\subsection{An \texorpdfstring{$L^{2}$}{L2}-mass expansion}

In the following subsections, we prove Proposition \ref{prop:MainProp}.
Our starting point is the following expansion.
\begin{lem}
Let $0<\delta<1$, $L=\lambda^{\delta}.$ We have
\begin{align}
 & \int_{B_{x}\left(r\right)}\left|g_{\lambda,L}\left(y\right)\right|^{2}\,dy-\pi r^{2}\nonumber \\
 & =\frac{2\pi r^{2}}{\left\Vert G_{\lambda,L}\right\Vert ^{2}}\sum_{\left|n-\lambda\right|<L}\frac{1}{\left(n-\lambda\right)^{2}}\sum_{\substack{\xi\ne\xi'\\
\left|\xi\right|^{2}=\left|\xi'\right|^{2}=n
}
}e^{i\left\langle x-x_{0},\xi-\xi'\right\rangle }\cdot\frac{J_{1}\left(r\left|\xi-\xi'\right|\right)}{r\left|\xi-\xi'\right|}\label{eq:MainFormula}\\
 & +\frac{2\pi r^{2}}{\left\Vert G_{\lambda,L}\right\Vert ^{2}}\sum_{\zeta\ne0}e^{i\left\langle x-x_{0},\zeta\right\rangle }\cdot\frac{J_{1}\left(r\left|\zeta\right|\right)}{r\left|\zeta\right|}\sum_{\begin{subarray}{c}
\left|\left|\xi\right|^{2}-\lambda\right|<L\\
\left|\left|\xi-\zeta\right|^{2}-\lambda\right|<L\\
\left|\xi\right|^{2}\ne\left|\xi-\zeta\right|^{2}
\end{subarray}}\frac{1}{\left(\left|\xi\right|^{2}-\lambda\right)\left(\left|\xi-\zeta\right|^{2}-\lambda\right)},\nonumber 
\end{align}
where $J_{1}$ is the Bessel function of the first kind of order $1$.
\end{lem}

\begin{proof}
We expand
\begin{align*}
 & \int_{B_{x}\left(r\right)}\left|g_{\lambda,L}\left(y\right)\right|^{2}\,dy=\frac{1}{\left\Vert G_{\lambda,L}\right\Vert ^{2}}\int_{B_{x}\left(r\right)}\sum_{\begin{subarray}{c}
\left|\left|\xi\right|^{2}-\lambda\right|<L\\
\left|\left|\xi'\right|^{2}-\lambda\right|<L
\end{subarray}}\frac{e^{i\left\langle y-x_{0},\xi-\xi'\right\rangle }}{\left(\left|\xi\right|^{2}-\lambda\right)\left(\left|\xi'\right|^{2}-\lambda\right)}\,dy\\
 & =\pi r^{2}+\frac{1}{\left\Vert G_{\lambda,L}\right\Vert ^{2}}\sum_{\begin{subarray}{c}
\left|\left|\xi\right|^{2}-\lambda\right|<L\\
\left|\left|\xi'\right|^{2}-\lambda\right|<L\\
\xi\ne\xi'
\end{subarray}}\frac{e^{i\left\langle -x_{0},\xi-\xi'\right\rangle }}{\left(\left|\xi\right|^{2}-\lambda\right)\left(\left|\xi'\right|^{2}-\lambda\right)}\int_{B_{x}\left(r\right)}e^{i\left\langle y,\xi-\xi'\right\rangle }\,dy.
\end{align*}
Note that by the change of variable $z=\frac{y-x}{r}$, we have
\[
\int_{B_{x}\left(r\right)}e^{i\left\langle y,\xi-\xi'\right\rangle }\,dy=r^{2}e^{i\left\langle x,\xi-\xi'\right\rangle }\int_{B_{0}\left(1\right)}e^{i\left\langle z,r\left(\xi-\xi'\right)\right\rangle }\,dz=2\pi r^{2}e^{i\left\langle x,\xi-\xi'\right\rangle }\frac{J_{1}\left(r\left|\xi-\xi'\right|\right)}{\left|\xi-\xi'\right|}
\]
where $J_{1}$ is the Bessel function of the first kind of order $1$.
Thus
\[
\int_{B_{x}\left(r\right)}\left|g_{\lambda,L}\left(y\right)\right|^{2}\,dy=\pi r^{2}+\frac{2\pi r^{2}}{\left\Vert G_{\lambda,L}\right\Vert ^{2}}\sum_{\begin{subarray}{c}
\left|\left|\xi\right|^{2}-\lambda\right|<L\\
\left|\left|\xi'\right|^{2}-\lambda\right|<L\\
\xi\ne\xi'
\end{subarray}}\frac{e^{i\left\langle x-x_{0},\xi-\xi'\right\rangle }}{\left(\left|\xi\right|^{2}-\lambda\right)\left(\left|\xi'\right|^{2}-\lambda\right)}\cdot\frac{J_{1}\left(r\left|\xi-\xi'\right|\right)}{r\left|\xi-\xi'\right|}.
\]
and (\ref{eq:MainFormula}) follows from the change of variable $\zeta=\xi-\xi'$.
\end{proof}

\subsection{The first term in (\ref{eq:MainFormula})}

We would like to evaluate the first term on the right-hand-side of
equation (\ref{eq:MainFormula}). For $\epsilon>0$, let

\[
BR\left(\epsilon\right)=\left\{ n\in\mathcal{N}_{2}:\,\min_{\substack{\xi\ne\xi'\\
\left|\xi\right|^{2}=\left|\xi'\right|^{2}=n
}
}\left| \xi-\xi'\right| \le n^{1/2-\epsilon}\right\} .
\]

Bourgain and Rudnick \cite{BourgainRudnick} proved that
\begin{equation}
\#\left\{ n\in BR\left(\epsilon\right):\,n\le X\right\} =O\left(X^{1-\epsilon/3}\right);\label{eq:BourgainRudnick}
\end{equation}
this bound was recently improved by Granville and Wigman \cite{GranvilleWigman}
to
\[
\#\left\{ n\in BR\left(\epsilon\right):\,n\le X\right\} =O\left(X^{1-\epsilon}\left(\log X\right)^{1/2}\right)
\]
by rather sophisticated methods, though the bound (\ref{eq:BourgainRudnick})
is sufficient for our application. On recalling (\ref{eq:LandauTheorem}),
it follows that $BR\left(\epsilon\right)$ is a density zero set in
$\mathcal{N}_{2}$, i.e., on most circles with radius $\sqrt{n}$,
there are no lattice points with distance smaller or equal to $n^{1/2-\epsilon}$.
The next lemma shows that most elements of $\Lambda$ are far from
the elements of $BR\left(\epsilon\right)$.
\begin{lem}
Let $\epsilon>0,$ $0<\delta<1$. We have
\begin{equation}
\#\left\{ \lambda\in\Lambda:\,\lambda\le X,\,\exists n\in BR\left(\epsilon\right).\,\left|n-\lambda\right|\le L\right\} =O\left(X^{1-\epsilon/3+\delta}\right)\label{eq:lambdaFarFromBR}
\end{equation}
where $L=\lambda^{\delta}.$
\end{lem}

\begin{proof}
We have
\[
\#\left\{ \lambda\in\Lambda:\,\lambda\le X,\,\exists n\in BR\left(\epsilon\right).\,\left|n-\lambda\right|\le L\right\} \le\sum_{\substack{\lambda\in\Lambda\\
\lambda\le X
}
}\sum_{\substack{n\in BR\left(\epsilon\right)\\
\left|n-\lambda\right|\le L
}
}1\le\sum_{\substack{n\le2X\\
n\in BR\left(\epsilon\right)
}
}\sum_{\substack{\lambda\in\Lambda\\
\left|n-\lambda\right|\le X^{\delta}
}
}1
\]
where in the second inequality we have changed the order of summation
and replaced $L$ with $X^{\delta}.$ The bound (\ref{eq:lambdaFarFromBR})
follows from the trivial bound
\[
\#\left\{ \lambda\in\Lambda:\,\left|n-\lambda\right|\le X^{\delta}\right\} \ll X^{\delta}
\]
and the bound (\ref{eq:BourgainRudnick}).
\end{proof}
Choosing $L=\lambda^{\delta}$ with $0<\delta=\delta\left(\epsilon\right)<\epsilon/3$,
we conclude that the set
\[
\Lambda''\left(\epsilon\right)=\left\{ \lambda\in\Lambda:\,\left|n-\lambda\right|\le L\Longrightarrow n\notin BR\left(\epsilon\right)\right\} 
\]
 is a density one subset in $\Lambda$. From now on, we restrict to
$\lambda\in\Lambda_{0}:=\Lambda''\left(\epsilon/2\right)$, so that
for $0<\delta=\delta\left(\epsilon\right)<\epsilon/6$, $\Lambda_{0}$
is a density one subset in $\Lambda$. We are now ready to evaluate
the first term on the right-hand-side of equation (\ref{eq:MainFormula}). 
\begin{lem}
Let $\epsilon>0$, $0<\delta<\epsilon/6$, $L=\lambda^{\delta}$.
We have
\begin{equation}
\sup_{\substack{x\in\mathbb{T}^{2}\\
r>\lambda^{-1/2+\epsilon}
}
}\left|\frac{1}{\left\Vert G_{\lambda,L}\right\Vert ^{2}}\sum_{\left|n-\lambda\right|<L}\frac{1}{\left(n-\lambda\right)^{2}}\sum_{\substack{\xi\ne\xi'\\
\left|\xi\right|^{2}=\left|\xi'\right|^{2}=n
}
}e^{i\left\langle x-x_{0},\xi-\xi'\right\rangle }\cdot\frac{J_{1}\left(r\left|\xi-\xi'\right|\right)}{r\left|\xi-\xi'\right|}\right|=O\left(\lambda^{-\epsilon/2}\right)\label{eq:SecondSumBound}
\end{equation}
along $\lambda\in\Lambda_{0}.$
\end{lem}

\begin{proof}
Since $\lambda\in\Lambda_{0}=\Lambda''\left(\epsilon/2\right),$ the
inner summation on the left-hand-side of (\ref{eq:SecondSumBound})
can be restricted to $\left|\xi-\xi'\right|>n^{1/2-\epsilon/2}.$
Since $J_{1}\left(x\right)\ll x^{-1/2},$ it follows that 
\begin{align*}
\sum_{\substack{\xi\ne\xi'\\
\left|\xi\right|^{2}=\left|\xi'\right|^{2}=n
}
}e^{i\left\langle x-x_{0},\xi-\xi'\right\rangle }\cdot\frac{J_{1}\left(r\left|\xi-\xi'\right|\right)}{r\left|\xi-\xi'\right|} & \ll\frac{1}{r^{3/2}}\sum_{\begin{subarray}{c}
\xi\ne\xi'\\
\left|\xi\right|^{2}=\left|\xi'\right|^{2}=n\\
\left|\xi-\xi'\right|>n^{1/2-\epsilon/2}
\end{subarray}.}\frac{1}{\left|\xi-\xi'\right|^{3/2}}\\
 & \ll\frac{1}{r^{3/2}}r_{2}\left(n\right)^{2}\frac{1}{n^{3/4-3\epsilon/4}}
\end{align*}
and therefore the expression inside the supremum in (\ref{eq:SecondSumBound})
is bounded above by
\begin{align*}
\frac{1}{\left\Vert G_{\lambda,L}\right\Vert ^{2}}\sum_{\left|n-\lambda\right|<L}\frac{1}{\left(n-\lambda\right)^{2}}\frac{r_{2}\left(n\right)^{2}}{r^{3/2}n^{3/4-3\epsilon/4}} & \ll\frac{\lambda^{-3/4+\epsilon}}{r^{3/2}}
\end{align*}
where we have used the bound (\ref{eq:r2nBound}) and the fact that
\[
\left\Vert G_{\lambda,L}\right\Vert ^{2}=\sum\limits _{\left|n-\lambda\right|<L}\frac{r_{2}\left(n\right)}{\left(n-\lambda\right)^{2}}.
\]
The estimate (\ref{eq:SecondSumBound}) now follows, since we take
the supremum over $r>\lambda^{-1/2+\epsilon}$.
\end{proof}

\subsection{The second term in (\ref{eq:MainFormula})}

Fix $\lambda\in\Lambda_{0}$. Our next goal is to evaluate the second
term on the right-hand-side of equation (\ref{eq:MainFormula}). Note
that
\[
\left|\left|\xi-\zeta\right|^{2}-\lambda\right|=\left|\left|\xi\right|^{2}-\lambda-2\left\langle \xi,\zeta\right\rangle +\left|\zeta\right|^{2}\right|
\]
so that
\begin{equation}
\left|\zeta\right|^{2}\le\left|\left|\xi-\zeta\right|^{2}-\lambda\right|+\left|\left|\xi\right|^{2}-\lambda\right|+2\left|\left\langle \xi,\zeta\right\rangle \right|.\label{eq:zetaBound}
\end{equation}
Thus, if $\left|\left|\xi\right|^{2}-\lambda\right|<L$ and $\left|\left|\xi-\zeta\right|^{2}-\lambda\right|<L$,
then using Cauchy-Schwarz inequality, (\ref{eq:zetaBound}) implies
\[
\left|\zeta\right|^{2}<2L+2\left|\left\langle \xi,\zeta\right\rangle \right|\le5\left|\zeta\right|\sqrt{\lambda}
\]
and therefore the outer summation of the second term in (\ref{eq:MainFormula})
can be restricted to $\left|\zeta\right|\le5\sqrt{\lambda}.$ We can
also write
\[
\left|\left|\xi-\zeta\right|^{2}-\lambda\right|=\left|\left|\xi\right|^{2}-\lambda-\left\langle 2\xi-\zeta,\zeta\right\rangle \right|,
\]
so that 
\[
\left|\left\langle 2\xi-\zeta,\zeta\right\rangle \right|\le\left|\left|\xi-\zeta\right|^{2}-\lambda\right|+\left|\left|\xi\right|^{2}-\lambda\right|,
\]
and therefore the inner summation of the second term in (\ref{eq:MainFormula})
can be restricted to $\xi\in\mathbb{Z}^{2}$ which satisfy $\left|\left|\xi\right|^{2}-\lambda\right|<L$
and
\[
0<\left|\left\langle 2\xi-\zeta,\zeta\right\rangle \right|<2L,
\]
and in particular (since $\left|\xi\right|^{2}\ge\lambda/2$ for $\lambda\ge4$,
assuming $\delta<\frac{1}{2}$), 
\[
0<\left|\left\langle 2\xi-\zeta,\zeta\right\rangle \right|\le3\left|\xi\right|^{2\delta}.
\]

For every $\zeta\ne0$, define
\[
A_{\zeta,\delta}=\left\{ \xi\in\mathbb{Z}^{2}:\,0<\left|\left\langle 2\xi-\zeta,\zeta\right\rangle \right|\le3\left|\xi\right|^{2\delta}\right\} 
\]
and
\[
B_{\zeta,\delta}=\left\{ \tilde{\xi}\in\mathbb{Z}^{2}:\,0<\left|\left\langle \tilde{\xi},\zeta\right\rangle \right|\le3\left|\frac{\tilde{\xi}+\zeta}{2}\right|^{2\delta}\right\} .
\]

Note that for any $ \xi \in A_{\zeta,\delta} $,  we have $ \tilde{\xi} := 2\xi - \zeta \in B_{\zeta,\delta} $, and $\left|\xi\right|^{2}\le2X $ implies $ \left|\tilde{\xi}+\zeta\right|^{2}\le8X $. Hence, for any $X$ sufficiently large such that $\left|\zeta\right|\le5\sqrt{X}$,
we have
\begin{align}
\#\left\{ \xi\in A_{\zeta,\delta}:\,\left|\xi\right|^{2}\le2X\right\}  & \le\#\left\{ \tilde{\xi}\in B_{\zeta,\delta}:\,\left|\tilde{\xi}+\zeta\right|^{2}\le8X\right\} \label{eq:AandBcompr}\\
 & \le\#\left\{ \tilde{\xi}\in B_{\zeta,\delta}:\,\left|\tilde{\xi}\right|\le\left|\zeta\right|+3X^{1/2}\right\} \nonumber \\
 & \le\#\left\{ \tilde{\xi}\in B_{\zeta,\delta}:\,\left|\tilde{\xi}\right|\le8X^{1/2}\right\} .\nonumber 
\end{align}
Next, we estimate the right-hand-side of (\ref{eq:AandBcompr}).
\begin{lem}
\label{lem:BUpperBound}Let $0<\delta<1/2$, $\zeta\ne0$. For $X$
sufficiently large such that $\left|\zeta\right|\le5\sqrt{X}$, we
have
\[
\#\left\{ \tilde{\xi}\in B_{\zeta,\delta}:\,\left|\tilde{\xi}\right|\le8X^{1/2}\right\} \ll\frac{X^{1/2+\delta}}{\left|\zeta\right|}.
\]
\end{lem}

\begin{proof}
Denote $\zeta=\left(m,n\right)$, and let $d=\gcd\left(m,n\right)$.
For any $l\in\mathbb{Z}$, there is a solution to the equation
\begin{equation}
\left\langle \tilde{\xi},\zeta\right\rangle =l\label{eq:linearDiophantineEq}
\end{equation}
if and only if $d\mid l$, and then if $\tilde{\xi}_{0}$ is a solution,
the other solutions are given by
\[
\tilde{\xi}=\tilde{\xi}_{0}+k\left(-\frac{n}{d},\frac{m}{d}\right)
\]
where $k\in\mathbb{Z}.$ Thus, for $l$ satisfying $d\mid l$, the
number of solutions $\tilde{\xi}$ to (\ref{eq:linearDiophantineEq})
such that $\left|\tilde{\xi}\right|\le8X^{1/2}$ is bounded above
(up to a constant factor) by
\[
\frac{X^{1/2}}{\left|\left(-\frac{n}{d},\frac{m}{d}\right)\right|}=\frac{dX^{1/2}}{\left|\zeta\right|}.
\]
We deduce that
\[
\#\left\{ \tilde{\xi}\in B_{\zeta,\delta}:\,\left|\tilde{\xi}\right|\le8X^{1/2}\right\} \ll\frac{dX^{1/2}}{\left|\zeta\right|}\sum_{\substack{1\le\left|l\right|\le21X^{\delta}\\
d\mid l
}
}1\ll\frac{X^{1/2+\delta}}{\left|\zeta\right|}
\]
(note that it is crucial that $l\ne0$ for the last inequality to
hold).
\end{proof}
Next, we show that for most of the elements of $\lambda\in\Lambda_{0}$,
if $\left|\left|\xi\right|^{2}-\lambda\right|\le L$, then $\xi$
is not an element of $A_{\zeta,\delta}.$
\begin{lem}
\label{lem:MainBoundLambdas}Let $0<\delta<1/2$, $\zeta\ne0$. For
$X$ sufficiently large such that $\left|\zeta\right|\le5\sqrt{X}$,
we have
\[
\#\left\{ \lambda\in\Lambda_{0}:\,\lambda\le X,\,\exists\xi\in A_{\zeta,\delta}.\,\left|\left|\xi\right|^{2}-\lambda\right|\le L\right\} \ll\frac{X^{1/2+2\delta}}{\left|\zeta\right|}
\]
where $L=\lambda^{\delta}.$
\end{lem}

\begin{proof}
We have
\begin{align*}
 \#\left\{ \lambda\in\Lambda_{0}:\,\lambda\le X,\,\exists\xi\in A_{\zeta,\delta}.\,\left|\left|\xi\right|^{2}-\lambda\right|\le L\right\} & \le\sum_{\substack{\lambda\in\Lambda_{0}\\
\lambda\le X
}
}\sum_{\begin{subarray}{c}
\xi\in A_{\zeta,\delta}\\
\left|\left|\xi\right|^{2}-\lambda\right|\le L
\end{subarray}}1 \\
& \le\sum_{\substack{\xi\in A_{\zeta,\delta}\\
\left|\xi\right|^{2}\le2X
}
}\sum_{\substack{\lambda\in\Lambda_{0}\\
\left|\left|\xi\right|^{2}-\lambda\right|\le X^{\delta}
}
}1,
\end{align*}
where in the last inequality we changed the order of summation and
replaced $L=\lambda^{\delta}$ with $X^{\delta}.$ By the trivial
bound 
\[
\left\{ \lambda\in\Lambda_{0}:\,\left|\left|\xi\right|^{2}-\lambda\right|\le X^{\delta}\right\} \ll X^{\delta}
\]
 and using the bound (\ref{eq:AandBcompr}) and Lemma \ref{lem:BUpperBound},
we get
\begin{align*}
\#\left\{ \lambda\in\Lambda_{0}:\,\lambda\le X,\,\exists\xi\in A_{\zeta,\delta}.\,\left|\left|\xi\right|^{2}-\lambda\right|\le L\right\}  & \ll X^{\delta}\cdot\#\left\{ \xi\in A_{\zeta,\delta}:\,\left|\xi\right|^{2}\le2X\right\} \\
 & \le X^{\delta}\cdot\#\left\{ \tilde{\xi}\in B_{\zeta,\delta}:\,\left|\tilde{\xi}\right|\le8X^{1/2}\right\} \\
 & \ll\frac{X^{1/2+2\delta}}{\left|\zeta\right|}.
\end{align*} 
\end{proof}
For $\lambda\in\Lambda_{0}$ and $\zeta\ne0$, denote 
\[
\chi_{\zeta}\left(\lambda\right)=\begin{cases}
1 & \exists\xi\in A_{\zeta,\delta}.\,\left|\left|\xi\right|^{2}-\lambda\right|\le L\\
0 & \text{otherwise}.
\end{cases}
\]

\begin{lem}
Let $\epsilon>0$, $0<\delta<\epsilon/6$, $\lambda\in\Lambda_{0},$
$L=\lambda^{\delta}$. We have
\begin{equation}
\sup_{\substack{x\in\mathbb{T}^{2}\\
r>\lambda^{-1/2+\epsilon}
}
}\left|\frac{1}{\pi r^{2}}\int_{B_{x}\left(r\right)}\left|g_{\lambda,L}\left(y\right)\right|^{2}\,dy-1\right|\ll\sum_{0<\left|\zeta\right|\le5\sqrt{\lambda}}\chi_{\zeta}\left(\lambda\right)\cdot\sup_{\substack{r>\lambda^{-1/2+\epsilon}}
}\frac{\left|J_{1}\left(r\left|\zeta\right|\right)\right|}{r\left|\zeta\right|}+\lambda^{-\epsilon/2}.\label{eq:MainBound}
\end{equation}
\end{lem}

\begin{proof}
By (\ref{eq:MainFormula}) and by the bound (\ref{eq:SecondSumBound}),
we have
\begin{align}
 & \left|\int_{B_{x}\left(r\right)}\left|g_{\lambda,L}\left(y\right)\right|^{2}\,dy-\pi r^{2}\right|\nonumber \\
 & \ll\frac{r^{2}}{\left\Vert G_{\lambda,L}\right\Vert ^{2}}\sum_{\zeta\ne0}\chi_{\zeta}\left(\lambda\right)\cdot\frac{\left|J_{1}\left(r\left|\zeta\right|\right)\right|}{r\left|\zeta\right|}\sum_{\begin{subarray}{c}
\left|\left|\xi\right|^{2}-\lambda\right|<L\\
\left|\left|\xi-\zeta\right|^{2}-\lambda\right|<L\\
\left|\xi\right|^{2}\ne\left|\xi-\zeta\right|^{2}
\end{subarray}}\frac{1}{\left|\left|\xi\right|^{2}-\lambda\right|\left|\left|\xi-\zeta\right|^{2}-\lambda\right|}+r^{2}\lambda^{-\epsilon/2}.\label{eq:BoundingIntegralWith_g_z}
\end{align}
By the Cauchy-Schwarz inequality, 
\begin{equation}
\sum_{\begin{subarray}{c}
\left|\left|\xi\right|^{2}-\lambda\right|<L\\
\left|\left|\xi-\zeta\right|^{2}-\lambda\right|<L\\
\left|\xi\right|^{2}\ne\left|\xi-\zeta\right|^{2}
\end{subarray}}\frac{1}{\left|\left|\xi\right|^{2}-\lambda\right|\left|\left|\xi-\zeta\right|^{2}-\lambda\right|}\le\left\Vert G_{\lambda,L}\right\Vert ^{2};\label{eq:CuachySchwarzG}
\end{equation}
the bound (\ref{eq:MainBound}) now follows substituting (\ref{eq:CuachySchwarzG})
into (\ref{eq:BoundingIntegralWith_g_z}).
\end{proof}

\subsection{Proof of Proposition \ref{prop:MainProp}}

Proposition \ref{prop:MainProp} will follow from the following estimate,
which we are now ready to prove.
\begin{lem}
\label{lem:MainBoundAverage}Let $\epsilon>0$, $0<\delta<\epsilon/6$.
We have
\begin{equation}
\sum_{\substack{\lambda\in\Lambda_{0}\\
\lambda\le X
}
}\sup_{\substack{x\in\mathbb{T}^{2}\\
r>\lambda^{-1/2+\epsilon}
}
}\left|\frac{1}{\pi r^{2}}\int_{B_{x}\left(r\right)}\left|g_{\lambda,L}\left(y\right)\right|^{2}\,dy-1\right|=O\left(X^{1-\epsilon/2}\right),\label{eq:MainBoundLambdaSum}
\end{equation}
where $L=\lambda^{\delta}.$ 
\end{lem}

\begin{proof}
By (\ref{eq:MainBound}), we have
\begin{align}
 & \sum_{\begin{subarray}{c}
\lambda\in\Lambda_{0}\\
\lambda\le X
\end{subarray}}\sup_{\substack{x\in\mathbb{T}^{2}\\
r>\lambda^{-1/2+\epsilon}
}
}\left|\frac{1}{\pi r^{2}}\int_{B_{x}\left(r\right)}\left|g_{\lambda,L}\left(y\right)\right|^{2}\,dy-1\right|\nonumber \\
 & \ll\sum_{\begin{subarray}{c}
\lambda\in\Lambda_{0}\\
\lambda\le X
\end{subarray}}\sum_{0<\left|\zeta\right|\le5\sqrt{\lambda}}\chi_{\zeta}\left(\lambda\right)\cdot\sup_{\substack{r>X^{-1/2+\epsilon}}
}\frac{\left|J_{1}\left(r\left|\zeta\right|\right)\right|}{r\left|\zeta\right|}+\sum_{\begin{subarray}{c}
\lambda\in\Lambda_{0}\\
\lambda\le X
\end{subarray}}\lambda^{-\epsilon/2}\nonumber \\
 & \text{\ensuremath{\ll}}\sum_{0<\left|\zeta\right|\le5\sqrt{X}}\#\left\{ \lambda\in\Lambda_{0}:\,\lambda\le X,\,\exists\xi\in A_{\zeta}.\,\left|\left|\xi\right|^{2}-\lambda\right|\le L\right\} \cdot\sup_{\substack{r>X^{-1/2+\epsilon}}
}\frac{\left|J_{1}\left(r\left|\zeta\right|\right)\right|}{r\left|\zeta\right|}\nonumber \\
 & +X^{1-\epsilon/2}\ll X^{1/2+2\delta}\sum_{\zeta\ne0}\sup_{\substack{r>X^{-1/2+\epsilon}}
}\frac{\left|J_{1}\left(r\left|\zeta\right|\right)\right|}{r\left|\zeta\right|^{2}}+X^{1-\epsilon/2}\label{eq:AverageIneq}
\end{align}
where the last estimate follows from Lemma \ref{lem:MainBoundLambdas}.
Since $J_{1}\left(x\right)\ll\min\left\{ x,x^{-1/2}\right\} ,$ we
have
\[
\sup_{\substack{r>X^{-1/2+\epsilon}}
}\frac{\left|J_{1}\left(r\left|\zeta\right|\right)\right|}{r\left|\zeta\right|^{2}}\ll\min\left\{ \frac{1}{\left|\zeta\right|},\frac{X^{3/4-\frac{3}{2}\epsilon}}{\left|\zeta\right|^{5/2}}\right\} .
\]
Thus 
\begin{equation}
\sum_{\zeta\ne0}\sup_{\substack{r>X^{-1/2+\epsilon}}
}\frac{\left|J_{1}\left(r\left|\zeta\right|\right)\right|}{r\left|\zeta\right|^{2}}\ll\sum_{0<\left|\zeta\right|\le X^{1/2-\epsilon}}\frac{1}{\left|\zeta\right|}+X^{3/4-\frac{3}{2}\epsilon}\sum_{\left|\zeta\right|>X^{1/2-\epsilon}}\frac{1}{\left|\zeta\right|^{5/2}}.\label{eq:BesselSum}
\end{equation}
We have (e.g., by comparison to an integral, or by partial summation)
\begin{equation}
\sum_{0<\left|\zeta\right|\le X^{1/2-\epsilon}}\frac{1}{\left|\zeta\right|}\ll X^{1/2-\epsilon}\label{eq:SmallZetasBound}
\end{equation}
and
\begin{equation}
X^{3/4-\frac{3}{2}\epsilon}\sum_{\left|\zeta\right|>X^{1/2-\epsilon}}\frac{1}{\left|\zeta\right|^{5/2}}\ll X^{1/2-\epsilon}.\label{eq:BigZetasBound}
\end{equation}
Substituting the bounds (\ref{eq:SmallZetasBound}) and (\ref{eq:BigZetasBound})
into (\ref{eq:BesselSum}) and then substituting back into (\ref{eq:AverageIneq}),
we get 
\begin{align*}
\sum_{\substack{\lambda\in\Lambda_{0}\\
\lambda\le X
}
}\sup_{\substack{x\in\mathbb{T}^{2}\\
r>\lambda^{-1/2+\epsilon}
}
}\left|\frac{1}{\pi r^{2}}\int_{B_{x}\left(r\right)}\left|g_{\lambda,L}\left(y\right)\right|^{2}\,dy-1\right| & \ll X^{1+2\delta-\epsilon}+X^{1-\epsilon/2}\ll X^{1-\epsilon/2}.
\end{align*}
\end{proof}
\begin{proof}[Proof of Proposition \ref{prop:MainProp}]
Let $0<\delta=\delta\left(\epsilon\right)<\epsilon/6$, and let
\[
B_{1}^{\epsilon}=\left\{ \lambda\in\Lambda_{0}:\sup_{\substack{x\in\mathbb{T}^{2}\\
r>\lambda^{-1/2+\epsilon}
}
}\left|\frac{1}{\pi r^{2}}\int_{B_{x}\left(r\right)}\left|g_{\lambda,L}\left(y\right)\right|^{2}\,dy-1\right|>\lambda^{-\epsilon/4}\right\} .
\]
Then by the Markov inequality and by Lemma \ref{lem:MainBoundAverage},
\begin{align*}
\#\left\{ \lambda\in B_{1}^{\epsilon}:\,\lambda\le X\right\}  & \ll X^{\epsilon/4}\sum_{\substack{\lambda\in\Lambda_{0}\\
\lambda\le X
}
}\sup_{\substack{x\in\mathbb{T}^{2}\\
r>\lambda^{-1/2+\epsilon}
}
}\left|\frac{1}{\pi r^{2}}\int_{B_{x}\left(r\right)}\left|g_{\lambda,L}\left(y\right)\right|^{2}\,dy-1\right|\ll X^{1-\epsilon/4}.
\end{align*}
Thus
\[
\Lambda_{1}^{\epsilon}:=\Lambda_{0}\setminus B_{1}^{\epsilon}=\left\{ \lambda\in\Lambda_{0}:\sup_{\substack{x\in\mathbb{T}^{2}\\
r>\lambda^{-1/2+\epsilon}
}
}\left|\frac{1}{\pi r^{2}}\int_{B_{x}\left(r\right)}\left|g_{\lambda,L}\left(y\right)\right|^{2}\,dy-1\right|<\lambda^{-\epsilon/4}\right\} 
\]
is a density one subset in $\Lambda$. By a standard diagonal argument,
we obtain a set $\Lambda_{1}$ of density one in $\Lambda$ which
does not depend on $\epsilon$, such that for every $\epsilon>0$
\[
\sup_{\substack{x\in\mathbb{T}^{2}\\
r>\lambda^{-1/2+\epsilon}
}
}\left|\frac{1}{\pi r^{2}}\int_{B_{x}\left(r\right)}\left|g_{\lambda,L}\left(y\right)\right|^{2}\,dy-1\right|\to0
\]
along $\lambda\in\Lambda_{1},$ completing the proof of Proposition
\ref{prop:MainProp}.
\end{proof}

\section{Proofs of Theorem \ref{thm:MainThm3d} and Theorem \ref{thm:Thm3dAE}
\texorpdfstring{($d=3$)}{(d=3)}}

\subsection{Preliminary lemmas}

Let $d=3$, and recall that $\mathcal{N}_{3}$ is the set of numbers
expressible as a sum of three squares, and that the elements of $\Lambda$
(the set of new eigenvalues) interlace with the elements of $\mathcal{N}_{3}$.
We recall the following results proved in \cite{Yesha1}, which are
the three-dimensional analogues of Lemma \ref{lem:LowerBound} and
Lemma \ref{lem:TruncationLemma}.
\begin{lem}[{\cite[Lemma 3.1, Lemma 3.2]{Yesha1}}]
~\label{lem:LowerBndG3d}

\begin{enumerate}

\item For all $\lambda\in\Lambda$ and all $ \eta>0 $, we have $\left\Vert G_{\lambda}\right\Vert ^{2}\gg_{\eta}\lambda^{1/2-\eta}.$

\item Let $0<\delta<1$. For all $\lambda\in\Lambda$ and all $ \eta>0 $, we have $\left\Vert g_{\lambda}-g_{\lambda,L}\right\Vert ^{2}\ll_{\eta}\lambda^{-\delta+\eta}$
($L=\lambda^{\delta})$.

\item For all $\lambda\in\Lambda$ and all $ \eta>0 $, we have $\left\Vert G_{\lambda,L}\right\Vert ^{2}\gg_{\eta}\lambda^{1/2-\eta}$
($L=\lambda^{\delta})$.

\end{enumerate}

\end{lem}

As a corollary of Lemma \ref{lem:LowerBndG3d}, we state the following
lemma, which is the three-dimensional analogue of Lemma \ref{lem:ApproximationLemma}.
\begin{lem}
\label{lem:ApproximationLemma3d}Let $0<\delta<1$. For every $\lambda\in\Lambda$ and every $ \eta>0 $,
we have 
\[
\sup_{\substack{x\in\mathbb{T}^{3}\\
r\in\mathbb{R}
}
}\left|\int_{B_{x}\left(r\right)}\left|g_{\lambda}\left(y\right)\right|^{2}\,dy-\int_{B_{x}\left(r\right)}\left|g_{\lambda,L}\left(y\right)\right|^{2}\,dy\right|=O_{\eta}\left(\lambda^{\eta-\delta/2}\right),
\]
where $L=\lambda^{\delta}.$
\end{lem}

\begin{proof}
Similar to the proof of Lemma \ref{lem:ApproximationLemma}.
\end{proof}
Taking Lemma \ref{lem:ApproximationLemma3d} into account, Theorem
\ref{thm:MainThm3d} immediately follows from the following proposition.
\begin{prop}
\label{prop:MainProp3d}Let $\epsilon>0$. We have
\[
\sup_{\substack{x\in\mathbb{T}^{2}\\
r>\lambda^{-1/12+\epsilon}
}
}\left|\frac{1}{\frac{4}{3}\pi r^{3}}\int_{B_{x}\left(r\right)}\left|g_{\lambda,L}\left(y\right)\right|^{2}\,dy-1\right|\to0
\]
as $\lambda\to\infty$ along $\Lambda$, where $L=\lambda^{\delta}$,
$0<\delta<\epsilon$. 
\end{prop}

\begin{proof}[Proof of Theorem \ref{thm:MainThm3d} assuming Proposition \ref{prop:MainProp3d}]
 By Lemma \ref{lem:ApproximationLemma3d}, for every $x\in\mathbb{T}^{3}$, $r>\lambda^{-1/12+\epsilon},$
we have
\begin{align*}
\left|\frac{1}{\frac{4}{3}\pi r^{3}}\int_{B_{x}\left(r\right)}\left|g_{\lambda}\left(y\right)\right|^{2}\,dy-1\right| & \le\left|\frac{1}{\frac{4}{3}\pi r^{3}}\int_{B_{x}\left(r\right)}\left|g_{\lambda,L}\left(y\right)\right|^{2}\,dy-1\right|+O_{\eta}\left(\lambda^{\eta-\delta/2}\right).
\end{align*}
Hence,
\begin{align*}
\sup_{\substack{x\in\mathbb{T}^{2}\\
r>\lambda^{-1/12+\epsilon}
}
}\left|\frac{1}{\frac{4}{3}\pi r^{3}}\int_{B_{x}\left(r\right)}\left|g_{\lambda}\left(y\right)\right|^{2}\,dy-1\right| & \le\sup_{\substack{x\in\mathbb{T}^{2}\\
r>\lambda^{-1/12+\epsilon}
}
}\left|\frac{1}{\frac{4}{3}\pi r^{3}}\int_{B_{x}\left(r\right)}\left|g_{\lambda,L}\left(y\right)\right|^{2}\,dy-1\right|\\
 & +O_{\eta}\left(\lambda^{\eta-\delta/2}\right).
\end{align*}
Theorem \ref{thm:MainThm3d} now follows from Proposition \ref{prop:MainProp3d},
choosing $\delta<\epsilon$ and $\eta<\delta/2$.
\end{proof}
We will use the following approximation to the characteristic function
of a ball of radius $r$ (``Beurling-Selberg polynomials'') appearing
in the work of Harman \cite{Harman}.
\begin{lem}[{\cite[Lemma 4]{Harman}}]
\label{lem:HarmanLemma}Let $T,r>0$ such that $Tr>1$. There exist
trigonometric polynomials $a^{\pm}$ such that:

\begin{enumerate}
	\item  $a^{-}\left(y\right)\le\mathbf{1}_{B_{0}\left(r\right)}\left(y\right)\le a^{+}\left(y\right).$
	
	\item  $\hat{a}^{\pm}\left(\zeta\right)=0$ if $\left|\zeta\right|\ge T$.
	
	\item $\hat{a}^{\pm}\left(0\right)=\mathrm{vol}\left(B_{0}\left(r\right)\right)+O_{d}\left(r^{d-1}/T\right).$
	
	\item  $\left|\hat{a}^{\pm}\left(\zeta\right)\right|\ll_{d}r^{d}.$
\end{enumerate}

\end{lem}

Given $x\in\mathbb{T}^{3}$, the polynomials $b_{x}^{\pm}\left(y\right)=a^{\pm}\left(y-x\right)$
satisfy 
\begin{equation}
b_{x}^{-}\left(y\right)\le\mathbf{1}_{B_{x}\left(r\right)}\left(y\right)\le b_{x}^{+}\left(y\right)\label{eq:Prop1b}
\end{equation}
and also satisfy properties (2)--(4) of Lemma \ref{lem:HarmanLemma}.
Proposition \ref{prop:MainProp3d} will follow from the following
proposition, which will be proved in the following subsections.
\begin{prop}
\label{prop:MainLemmaSmoothed} Let $\epsilon>0$. We have
\end{prop}

\[
\sup_{\substack{x\in\mathbb{T}^{2}\\
r>\lambda^{-1/12+\epsilon}
}
}\left|\frac{1}{\frac{4}{3}\pi r^{3}}\int_{\mathbb{T}^{3}}b_{x}^{\pm}\left(y\right)\left|g_{\lambda,L}\left(y\right)\right|^{2}\,dy-1\right|\to0
\]
as $\lambda\to\infty$ along $\Lambda$, where $L=\lambda^{\delta}$,
$0<\delta<\epsilon$.
\begin{proof}[Proof of Proposition \ref{prop:MainProp3d} assuming Proposition \ref{prop:MainLemmaSmoothed}]
 For every $x\in\mathbb{T}^{3},\,r>\lambda^{-1/12+\epsilon},$ the
bounds (\ref{eq:Prop1b}) imply
\begin{align*}
\frac{1}{\frac{4}{3}\pi r^{3}}\int_{\mathbb{T}^{3}}b_{x}^{-}\left(y\right)\left|g_{\lambda,L}\left(y\right)\right|^{2}\,dy & \le\frac{1}{\frac{4}{3}\pi r^{3}}\int_{B_{x}\left(r\right)}\left|g_{\lambda,L}\left(y\right)\right|^{2}\,dy\\
 & \le\frac{1}{\frac{4}{3}\pi r^{3}}\int_{\mathbb{T}^{3}}b_{x}^{+}\left(y\right)\left|g_{\lambda,L}\left(y\right)\right|^{2}\,dy.
\end{align*}
Proposition \ref{prop:MainProp3d} now clearly follows from Proposition
\ref{prop:MainLemmaSmoothed}.
\end{proof}

\subsection{An \texorpdfstring{$L^{2}$}{L2}-mass expansion}

Our starting point towards proving Proposition \ref{prop:MainLemmaSmoothed},
as well as proving Theorem \ref{thm:Thm3dAE}, is the following expansion.

\begin{lem}
\label{lem:ExpansionLemma}Let $T,r>0$ such that $Tr>1$, and let
$0<\delta<1$, $L=\lambda^{\delta}.$ We have
\begin{align}
 & \int_{\mathbb{T}^{3}}b_{x}^{\pm}\left(y\right)\left|g_{\lambda,L}\left(y\right)\right|^{2}\,dy-\frac{4}{3}\pi r^{3}\nonumber \\
 & =\frac{1}{\left\Vert G_{\lambda,L}\right\Vert ^{2}}\sum_{0<\left|\zeta\right|<T}\hat{b}_{x}^{\pm}\left(\zeta\right)e^{-i\left\langle x_{0},\zeta\right\rangle }\sum_{\begin{subarray}{c}
\left|\left|\xi\right|^{2}-\lambda\right|<L\\
\left|\left|\xi-\zeta\right|^{2}-\lambda\right|<L
\end{subarray}}\frac{1}{\left(\left|\xi\right|^{2}-\lambda\right)\left(\left|\xi-\zeta\right|^{2}-\lambda\right)}+O\left(r^{2}/T\right).\label{eq:3dFormulaAfterExpansion}
\end{align}
\end{lem}

\begin{proof}
We expand
\[
\int_{\mathbb{T}^{3}}b_{x}^{\pm}\left(y\right)\left|g_{\lambda,L}\left(y\right)\right|^{2}\,dy=\frac{1}{\left\Vert G_{\lambda,L}\right\Vert ^{2}}\sum_{\begin{subarray}{c}
\left|\left|\xi\right|^{2}-\lambda\right|<L\\
\left|\left|\xi'\right|^{2}-\lambda\right|<L
\end{subarray}}\frac{e^{-i\left\langle x_{0},\xi-\xi'\right\rangle }}{\left(\left|\xi\right|^{2}-\lambda\right)\left(\left|\xi'\right|^{2}-\lambda\right)}\hat{b}_x^\pm(\xi-\xi')
\]
and writing $\zeta=\xi-\xi'$, the lemma follows from properties (2)
and (3) of Lemma \ref{lem:HarmanLemma}.
\end{proof}
For each $\zeta\ne0$, denote $\left|\zeta\right|^{2}=n_{\zeta}=4^{a_{\zeta}}n_{1}^{\zeta}$
with $4\nmid n_{1}^{\zeta}$.
\begin{defn}
Define
\[
\mathcal{N}_{0}^{\zeta}=\left\{ n\in\mathcal{N}_{3}:\,n=4^{a}n_{1},\,4\nmid n_{1}\Rightarrow a>a_{\zeta}\right\} 
\]
and
\[
\mathcal{N}_{1}^\zeta=\left\{ n\in\mathcal{N}_{3}:\,n=4^{a}n_{1},\,4\nmid n_{1}\Rightarrow a\le a_{\zeta}\right\} .
\]
Denote by $n_{\lambda}$ the element of $\mathcal{N}_{3}$ which is
closest to $\lambda\in\Lambda$ (if there are two elements with the
same distance, we take the $n_{\lambda}$ to be the smallest among
them). We recall the following results, proved in \cite{Yesha1}.
\end{defn}

\begin{lem}[{\cite[Lemma 3.6, Corollary 3.7, Lemma 3.8]{Yesha1}}]
\label{lem:LemmaN_0}~

\begin{enumerate}
	\item  For every $\xi\in\mathbb{Z}^{3}$, if $\left|\xi\right|^{2}=\left|\xi-\zeta\right|^{2},$
	then $\left|\xi\right|^{2}\in\mathcal{N}_{1}^{\zeta}.$ 
	
	\item For every $\xi\in\mathbb{Z}^{3}$, if $\left|\xi\right|^{2}\in\mathcal{N}_{0}^{\zeta}$,
	then $\left|\left|\xi\right|^{2}-\left|\xi-\zeta\right|^{2}\right|\ge1$.
	
	\item If $n_{\lambda}\in\mathcal{N}_{0}^{\zeta}$, then $\left|\left|\xi\right|^{2}-\lambda\right|<1/2\Rightarrow\left|\left|\xi-\zeta\right|^{2}-\lambda\right|>1/2$,
	for every $\xi\in\mathbb{Z}^{3}$.
\end{enumerate}

\end{lem}

To evaluate the right-hand-side of (\ref{eq:3dFormulaAfterExpansion}),
we consider separately the summation over $0<\left|\zeta\right|<T$
such that $n_{\lambda}\in\mathcal{N}_{0}^{\zeta}$, and over $0<\left|\zeta\right|<T$
such that $n_{\lambda}\in\mathcal{N}_{1}^{\zeta}$. Our main tool
for evaluating the sums in both cases, will be the following estimate
for the number of lattice points inside spherical strips, which is
a uniform version of \cite[Lemma A.1]{Yesha1}.
\begin{lem}
\label{lem:SphericalStripsLemma}Let $C_{1},C_{2},C_{3}>0$, $0<\delta<1$,
$L=\lambda^{\delta},$ $0<\left|\zeta\right|\le C_{1}\sqrt{\lambda}$.
For every $n\in\mathcal{N}_{3}$ satisfying $\left|n-\lambda\right|<C_{2}L$ and every $ \eta>0 $,
we have
\[
\#\left\{ \xi\in\mathbb{Z}^{3}:\,\left|\xi\right|^{2}=n,\,\left|\left\langle 2\xi-\zeta,\zeta\right\rangle \right|<C_{3}L\right\} \ll_{C_{1},C_{2},C_{3},\eta}L\lambda^{\eta}.
\]
\end{lem}

\begin{proof}
Denote $\zeta=\left(\zeta_{1},\zeta_{2},\zeta_{3}\right),$ and assume
without loss of generality that $\zeta_{3}\ne0$. Let $n\in\mathcal{N}_{3}$
so that $\left|n-\lambda\right|<C_{2}L$, so we are looking for $\xi=\left(x,y,z\right)$
such $\left|\xi\right|^{2}=n,$ and such that $\left\langle 2\xi-\zeta,\zeta\right\rangle =m$,
where $\left|m\right|<C_{3}L$. In this notation, we have 
\begin{equation}
x^{2}+y^{2}+z^{2}=n\label{eq:Triplet}
\end{equation}
and 
\[
2x\zeta_{1}+2y\zeta_{2}+2z\zeta_{3}=m+\left|\zeta\right|^{2}.
\]
Since $\zeta_{3}\ne0$, we can write
\begin{equation}
z=\frac{m+\left|\zeta\right|^{2}-2\zeta_{1}x-2\zeta_{2}y}{2\zeta_{3}}, \label{eq:Z_equation}
\end{equation}
and substituting (\ref{eq:Z_equation}) into (\ref{eq:Triplet}), we get
\begin{equation}
ax^{2}+2bxy+cy^{2}+2dx+2ey+f=0\label{eq:EllipseEquationPreSimplification}
\end{equation}
where
\begin{align*}
a & =4\zeta_{1}^{2}+4\zeta_{3}^{2}\\
b & =4\zeta_{1}\zeta_{2}\\
c & =4\zeta_{2}^{2}+4\zeta_{3}^{2}\\
d & =-2\zeta_{1}\left(m+\left|\zeta\right|^{2}\right)\\
e & =-2\zeta_{2}\left(m+\left|\zeta\right|^{2}\right)\\
f & =-4\zeta_{3}^{2}n+\left(m+\left|\zeta\right|^{2}\right)^{2}.
\end{align*}
Note that $c>0$, $ac-b^{2}=16\zeta_{3}^{2}\left|\zeta\right|^{2}>0$,
and denote $ac-b^{2}=t^{2}D$, where $D>0$ is squarefree. By a simple
sequence of changes of variables (see the proof of \cite[Lemma A.1]{Yesha1}),
the number of integer solutions to equation (\ref{eq:EllipseEquationPreSimplification})
is bounded above by the number of integer solutions to the equation
\begin{equation}
x^{2}+Dy^{2}=k,\label{eq:EllipseEquationSimplified}
\end{equation}
where 
\[
k=\left(ac-b^{2}\right)\left(-cf+e^{2}\right)+\left(cd-be\right)^{2}.
\]
In the proof of \cite[Lemma A.1]{Yesha1}), we obtained that the number
of integer solutions to equation (\ref{eq:EllipseEquationSimplified})
is bounded above by $6\tau\left(k\right)$, where $\tau\left(k\right)\ll_{\eta}k^{\eta}$
is the number of divisors of $k$. Since $\left|\zeta\right|\le C_{1}\sqrt{\lambda}$
and $\left|n-\lambda\right|<C_{2}L$, we deduce that the number of
solutions is $\ll_{C_{1},C_{2},C_{3},\eta}\lambda^{\eta}$ for every
fixed $m$ such that $\left|m\right|<C_{3}L$. Taking into account
the various choices for $m$ we get that the number of solutions is
$\ll_{C_{1},C_{2},C_{3},\eta}L\lambda^{\eta}$.
\end{proof}

\subsection{The case \texorpdfstring{$n_{\lambda}\in\mathcal{N}_{0}^{\zeta}$}{n\textlambda \textin N0\textzeta}}

In light of Lemma \ref{lem:LemmaN_0}, for every $0<\left|\zeta\right|<T$
such that $n_{\lambda}\in\mathcal{N}_{0}^{\zeta}$, we can split the
inner summation on the right-hand-side of (\ref{eq:3dFormulaAfterExpansion})
into three sums over the following ranges:

\begin{enumerate}
	\item $\frac{1}{2}\le\left|\left|\xi\right|^{2}-\lambda\right|<L$,
	$\frac{1}{2}\le\left|\left|\xi-\zeta\right|^{2}-\lambda\right|\le L$
	(Denote this sum by $\sum^{1}$).
	
	\item  $\left|\left|\xi\right|^{2}-\lambda\right|<\frac{1}{2}$, $\frac{1}{2}\le\left|\left|\xi-\zeta\right|^{2}-\lambda\right|\le L$
	(Denote this sum by $\sum^{2}$).

	\item $\frac{1}{2}\le\left|\left|\xi\right|^{2}-\lambda\right|<L$,
	$\left|\left|\xi-\zeta\right|^{2}-\lambda\right|<\frac{1}{2}$ (Denote
	this sum by $\sum^{3}$).
\end{enumerate}

We will use Lemma \ref{lem:SphericalStripsLemma} to estimate $\sum^{1},\sum^{2}$
and $\sum^{3}$, which will give the required bound in the case $n_{\lambda}\in\mathcal{N}_{0}^{\zeta}$.
\begin{lem}
\label{lem:MianBoundN0}Let $0<\delta<1$, $L=\lambda^{\delta}.$
For every fixed $0<\left|\zeta\right|<T$ such that $n_{\lambda}\in\mathcal{N}_{0}^{\zeta}$ and every $ \eta>0 $,
we have
\[
\sum_{\begin{subarray}{c}
\left|\left|\xi\right|^{2}-\lambda\right|<L\\
\left|\left|\xi-\zeta\right|^{2}-\lambda\right|<L
\end{subarray}}\frac{1}{\left(\left|\xi\right|^{2}-\lambda\right)\left(\left|\xi-\zeta\right|^{2}-\lambda\right)}\ll_{\eta}\lambda^{2\delta+\eta}+\left\Vert G_{\lambda,L}\right\Vert \lambda^{\delta/2+\eta/2}.
\]
\end{lem}

\begin{proof}
Recall that if $\left|\left|\xi\right|^{2}-\lambda\right|<L$ and
$\left|\left|\xi-\zeta\right|^{2}-\lambda\right|<L$, then
\[
\left|\left\langle 2\xi-\zeta,\zeta\right\rangle \right|=\left|\left|\xi\right|^{2}-\left|\xi-\zeta\right|^{2}\right|\le\left|\left|\xi\right|^{2}-\lambda\right|+\left|\left|\xi-\zeta\right|^{2}-\lambda\right|<2L.
\]
By Lemma \ref{lem:SphericalStripsLemma}, we have
\begin{align*}
\sum\nolimits ^{1}\frac{1}{\left(\left|\xi\right|^{2}-\lambda\right)\left(\left|\xi-\zeta\right|^{2}-\lambda\right)} & \ll\sum\nolimits ^{1}1\\
 & \ll\sum_{\left|n-\lambda\right|<L}\#\left\{ \xi\in\mathbb{Z}^{3}:\,\left|\xi\right|^{2}=n,\,\left|\left\langle 2\xi-\zeta,\zeta\right\rangle \right|<2L\right\} \\
 & \ll_{\eta}L\lambda^{\eta}\sum_{\left|n-\lambda\right|<L}1\ll L^{2}\lambda^{\eta}=\lambda^{2\delta+\eta}.
\end{align*}
Next, using the Cauchy-Schwarz inequality and Lemma \ref{lem:SphericalStripsLemma}
again, we have
\begin{align*}
 & \sum\nolimits ^{2}\frac{1}{\left(\left|\xi\right|^{2}-\lambda\right)\left(\left|\xi-\zeta\right|^{2}-\lambda\right)}\ll\sum\nolimits ^{2}\frac{1}{\left|\left|\xi\right|^{2}-\lambda\right|}\ll\left\Vert G_{\lambda,L}\right\Vert \left(\sum\nolimits ^{2}1\right)^{1/2}\\
 & \ll\left\Vert G_{\lambda,L}\right\Vert \left(\#\left\{ \xi\in\mathbb{Z}^{3}:\,\left|\xi\right|^{2}=n_{\lambda},\,\left|\left\langle 2\xi-\zeta,\zeta\right\rangle \right|<2L\right\} \right)^{1/2}\\
 & \ll_{\eta}\left\Vert G_{\lambda,L}\right\Vert L^{1/2}\lambda^{\eta/2}=\left\Vert G_{\lambda,L}\right\Vert \lambda^{\delta/2+\eta/2}.
\end{align*}
The sum $\sum^{3}$ is treated like $\sum^{2}$ and is bounded by
the same quantity. 
\end{proof}
We deduce a bound for the sum on the right-hand-side of (\ref{eq:3dFormulaAfterExpansion})
restricted to $0<\left|\zeta\right|<T$ such that $n_{\lambda}\in\mathcal{N}_{0}^{\zeta}$.
\begin{cor}
Let $T,r>0$ such that $Tr>1$, and let $0<\delta<1$, $L=\lambda^{\delta}.$ For every $ \eta>0 $,
we have
\begin{align}
 & \sum_{\begin{subarray}{c}
0<\left|\zeta\right|<T\\
n_{\lambda}\in\mathcal{N}_{0}^{\zeta}
\end{subarray}}\hat{b}_{x}^{\pm}\left(\zeta\right)e^{-i\left\langle x_{0},\zeta\right\rangle }\sum_{\begin{subarray}{c}
\left|\left|\xi\right|^{2}-\lambda\right|<L\\
\left|\left|\xi-\zeta\right|^{2}-\lambda\right|<L
\end{subarray}}\frac{1}{\left(\left|\xi\right|^{2}-\lambda\right)\left(\left|\xi-\zeta\right|^{2}-\lambda\right)}\label{eq:N0FinalBound}\\
 & \ll_{\eta}r^{3}T^{3}\left(\lambda^{2\delta+\eta}+\left\Vert G_{\lambda,L}\right\Vert \lambda^{\delta/2+\eta/2}\right).\nonumber 
\end{align}
\end{cor}

\begin{proof}
This follows immediately from Lemma \ref{lem:MianBoundN0} and property
(4) of Lemma \ref{lem:HarmanLemma}. 
\end{proof}

\subsection{The case \texorpdfstring{$n_{\lambda}\in\mathcal{N}_{1}^{\zeta}$}{n\textlambda \textin N1\textzeta}}
\begin{lem}
\label{lem:N1Lemma}Let $0<\delta<1$, $L=\lambda^{\delta}.$ For
every fixed $0<\left|\zeta\right|<T$ such that $n_{\lambda}\in\mathcal{N}_{1}^{\zeta}$ and every $ \eta>0 $,
we have
\[
\sum_{\begin{subarray}{c}
\left|\left|\xi\right|^{2}-\lambda\right|<L\\
\left|\left|\xi-\zeta\right|^{2}-\lambda\right|<L
\end{subarray}}\frac{1}{\left(\left|\xi\right|^{2}-\lambda\right)\left(\left|\xi-\zeta\right|^{2}-\lambda\right)}\ll_{\eta}\lambda^{2\delta+\eta}+\left\Vert G_{\lambda,L}\right\Vert ^{2}T\lambda^{\delta-\frac{1}{2}+2\eta}.
\]
\end{lem}

\begin{proof}
By the Cauchy-Schwarz inequality,
\begin{align}
 & \sum_{\begin{subarray}{c}
\left|\left|\xi\right|^{2}-\lambda\right|<L\\
\left|\left|\xi-\zeta\right|^{2}-\lambda\right|<L
\end{subarray}}\frac{1}{\left(\left|\xi\right|^{2}-\lambda\right)\left(\left|\xi-\zeta\right|^{2}-\lambda\right)}\label{eq:N1Bound}\\
 & \ll\left(\sum_{\begin{subarray}{c}
\left|\left|\xi\right|^{2}-\lambda\right|<L\\
\left|\left\langle 2\xi-\zeta,\zeta\right\rangle \right|<2L
\end{subarray}}\frac{1}{\left(\left|\xi\right|^{2}-\lambda\right)^{2}}\right)^{1/2}\left(\sum_{\begin{subarray}{c}
\left|\left|\xi-\zeta\right|^{2}-\lambda\right|<L\\
\left|\left\langle 2\xi-\zeta,\zeta\right\rangle \right|<2L
\end{subarray}}\frac{1}{\left(\left|\xi-\zeta\right|^{2}-\lambda\right)^{2}}\right)^{1/2}.\nonumber 
\end{align}
Consider the first summation on the right-hand-side of (\ref{eq:N1Bound})
\begin{equation}
\sum_{\begin{subarray}{c}
\left|\left|\xi\right|^{2}-\lambda\right|<L\\
\left|\left\langle 2\xi-\zeta,\zeta\right\rangle \right|<2L
\end{subarray}}\frac{1}{\left(\left|\xi\right|^{2}-\lambda\right)^{2}}\label{eq:FirstSumN1}
\end{equation}
and split it into two sums $\sum^{4}+\sum^{5}$, where in $\sum^{4}$
the summation is over $\xi$ such that $\left|\xi\right|^{2}\ne n_{\lambda}$,
$\left|\left|\xi\right|^{2}-\lambda\right|<L$ and $\left|\left\langle 2\xi-\zeta,\zeta\right\rangle \right|<2L$,
and in $\sum^{5}$ the summation is over $\xi$ such that $\left|\xi\right|^{2}=n_{\lambda}$
and $\left|\left\langle 2\xi-\zeta,\zeta\right\rangle \right|<2L$.

If $\left|\xi\right|^{2}\ne n_{\lambda}$, then $\left|\left|\xi\right|^{2}-\lambda\right|\ge\frac{1}{2},$
and therefore Lemma \ref{lem:SphericalStripsLemma} yields
\begin{align}
\sum\nolimits ^{4}\frac{1}{\left(\left|\xi\right|^{2}-\lambda\right)^{2}} & \ll\sum_{\begin{subarray}{c}
\left|\left|\xi\right|^{2}-\lambda\right|<L\\
\left|\left\langle 2\xi-\zeta,\zeta\right\rangle \right|<2L
\end{subarray}}1\ll\sum_{\left|n-\lambda\right|<L}\#\left\{ \xi\in\mathbb{Z}^{3}:\,\left|\xi\right|^{2}=n,\,\left|\left\langle 2\xi-\zeta,\zeta\right\rangle \right|<2L\right\} \label{eq:Sigma4Bound}\\
 & \ll_{\eta}\sum_{\left|n-\lambda\right|<L}L\lambda^{\eta}\ll L^{2}\lambda^{\eta}=\lambda^{2\delta+\eta}.\nonumber 
\end{align}
To bound $\sum^{5}$, note that by Lemma \ref{lem:SphericalStripsLemma},
\begin{align*}
\sum\nolimits ^{5}\frac{1}{\left(\left|\xi\right|^{2}-\lambda\right)^{2}} & =\frac{1}{\left(n_{\lambda}-\lambda\right)^{2}}\#\left\{ \xi\in\mathbb{Z}^{3}:\,\left|\xi\right|^{2}=n_{\lambda},\,\left|\left\langle 2\xi-\zeta,\zeta\right\rangle \right|<2L\right\} \\
& \ll_{\eta}\frac{L\lambda^{\eta}}{\left(n_{\lambda}-\lambda\right)^{2}}
\end{align*}
and therefore
\begin{equation}
\frac{\sum^{5}\frac{1}{\left(\left|\xi\right|^{2}-\lambda\right)^{2}}}{\left\Vert G_{\lambda,L}\right\Vert ^{2}}\ll_{\eta}\frac{\frac{L\lambda^{\eta}}{\left(n_{\lambda}-\lambda\right)^{2}}}{\left\Vert G_{\lambda,L}\right\Vert ^{2}}\ll\frac{\frac{L\lambda^{\eta}}{\left(n_{\lambda}-\lambda\right)^{2}}}{\frac{r_{3}\left(n_{\lambda}\right)}{\left(n_{\lambda}-\lambda\right)^{2}}}=\frac{L\lambda^{\eta}}{r_{3}\left(n_{\lambda}\right)}.\label{eq:Sigma5Bound}
\end{equation}
But since $n_{\lambda}\in\mathcal{N}_{1}^{\zeta}$, if we write $n_{\lambda}=4^{a}n_{1}$
with $4\nmid n_{1}$, then $n_{\lambda}\le4^{a_{\zeta}}n_{1}$. Moreover,
$2^{a_{\zeta}}\le T$, and therefore, by (\ref{eq:SiegelBound}),
\[
\frac{L\lambda^{\eta}}{r_{3}\left(n_{\lambda}\right)}=\frac{\lambda^{\delta+\eta}}{r_{3}\left(n_{1}\right)}\ll_{\eta}\frac{\lambda^{\delta+\eta}}{n_{1}^{1/2-\eta}}\ll\frac{2^{a_{\zeta}}\lambda^{\delta+\eta}}{n_{\lambda}^{1/2-\eta}}\ll T\lambda^{\delta-\frac{1}{2}+2\eta},
\]
so
\[
\sum\nolimits ^{5}\frac{1}{\left(\left|\xi\right|^{2}-\lambda\right)^{2}}\ll_{\eta}\left\Vert G_{\lambda,L}\right\Vert ^{2}T\lambda^{\delta-\frac{1}{2}+2\eta}.
\]
As for the second summation on the right-hand-side of (\ref{eq:N1Bound}),
note that
\begin{equation}
\sum_{\begin{subarray}{c}
\left|\left|\xi-\zeta\right|^{2}-\lambda\right|<L\\
\left|\left\langle 2\xi-\zeta,\zeta\right\rangle \right|<2L
\end{subarray}}\frac{1}{\left(\left|\xi-\zeta\right|^{2}-\lambda\right)^{2}}=\sum_{\begin{subarray}{c}
\left|\left|\tilde{\xi}\right|^{2}-\lambda\right|<L\\
\left|\left\langle 2\tilde{\xi}+\zeta,\zeta\right\rangle \right|<2L
\end{subarray}}\frac{1}{\left(\left|\tilde{\xi}\right|^{2}-\lambda\right)^{2}}\label{eq:SecondSumN1}
\end{equation}
 and this summation can be bounded in the same way. Therefore
\[
\sum_{\begin{subarray}{c}
\left|\left|\xi\right|^{2}-\lambda\right|<L\\
\left|\left|\xi-\zeta\right|^{2}-\lambda\right|<L
\end{subarray}}\frac{1}{\left(\left|\xi\right|^{2}-\lambda\right)\left(\left|\xi-\zeta\right|^{2}-\lambda\right)}\ll_{\eta}\lambda^{2\delta+\eta}+\left\Vert G_{\lambda,L}\right\Vert ^{2}T\lambda^{\delta-\frac{1}{2}+2\eta}.
\]
\end{proof}
We deduce a bound for the sum on the right-hand-side of (\ref{eq:3dFormulaAfterExpansion}),
now restricted to $0<\left|\zeta\right|<T$ such that $n_{\lambda}\in\mathcal{N}_{1}^{\zeta}$.
\begin{cor}
Let $T,r>0$ such that $Tr>1$, and let $0<\delta<1$, $L=\lambda^{\delta}.$ For every $ \eta>0 $, we have
\begin{align}
 & \sum_{\begin{subarray}{c}
0<\left|\zeta\right|<T\\
n_{\lambda}\in\mathcal{N}_{1}^{\zeta}
\end{subarray}}\hat{b}_{x}^{\pm}\left(\zeta\right)e^{-i\left\langle x_{0},\zeta\right\rangle }\sum_{\begin{subarray}{c}
\left|\left|\xi\right|^{2}-\lambda\right|<L\\
\left|\left|\xi-\zeta\right|^{2}-\lambda\right|<L
\end{subarray}}\frac{1}{\left(\left|\xi\right|^{2}-\lambda\right)\left(\left|\xi-\zeta\right|^{2}-\lambda\right)}\label{eq:N1FinalBound}\\
 & \ll_{\eta}r^{3}T^{3}\left(\lambda^{2\delta+\eta}+\left\Vert G_{\lambda,L}\right\Vert ^{2}T\lambda^{\delta-\frac{1}{2}+2\eta}\right).\nonumber 
\end{align}
\end{cor}

\begin{proof}
This follows immediately from Lemma \ref{lem:N1Lemma} and property
(4) of Lemma \ref{lem:HarmanLemma}. 
\end{proof}

\subsection{Proof of Proposition \ref{prop:MainLemmaSmoothed}}

We are now able to prove Proposition \ref{prop:MainLemmaSmoothed}.
\begin{proof}[Proof of Proposition \ref{prop:MainLemmaSmoothed} ]
Substituting (\ref{eq:N0FinalBound}) and (\ref{eq:N1FinalBound})
into the right-hand-side of (\ref{eq:3dFormulaAfterExpansion}) and
using the third part of Lemma \ref{lem:LowerBndG3d}, we conclude
that for $r,T>0$ such that $Tr>1$, we have
\begin{align*}
\int_{\mathbb{T}^{3}}b_{x}^{\pm}\left(y\right)\left|g_{\lambda,L}\left(y\right)\right|^{2}\,dy-\frac{4}{3}\pi r^{3} & \ll_{\eta}r^{3}T^{3}\left(\frac{\lambda^{2\delta+\eta}}{\left\Vert G_{\lambda,L}\right\Vert ^{2}}+\frac{\lambda^{\frac{\delta}{2}+\frac{\eta}{2}}}{\left\Vert G_{\lambda,L}\right\Vert }+T\lambda^{\delta-\frac{1}{2}+2\eta}\right)+\frac{r^{2}}{T}\\
 & \ll r^{3}T^{4}\lambda^{\delta-\frac{1}{2}+2\eta}+r^{3}T^{3}\left(\lambda^{2\delta-\frac{1}{2}+2\eta}+\lambda^{\frac{\delta}{2}-\frac{1}{4}+\eta}\right)+\frac{r^{2}}{T}.
\end{align*}
 Choose $T=\lambda^{1/12-\frac{\epsilon}{2}}.$ If $r>\lambda^{-1/12+\epsilon},$
then $Tr>1$, and hence if we take $\eta=\epsilon$, 
\[
\int_{\mathbb{T}^{3}}b_{x}^{\pm}\left(y\right)\left|g_{\lambda,L}\left(y\right)\right|^{2}\,dy-\frac{4}{3}\pi r^{3}\ll r^{3}\left(\lambda^{\frac{\delta}{2}-\frac{\epsilon}{2}}+\frac{1}{r\lambda^{\frac{1}{12}-\frac{\epsilon}{2}}}\right)\ll r^{3}\left(\lambda^{\frac{\delta}{2}-\frac{\epsilon}{2}}+\lambda^{-\frac{\epsilon}{2}}\right)
\]
and the statement of the proposition follows by choosing $0<\delta<\epsilon$.
\end{proof}

\subsection{Proof of Theorem \ref{thm:Thm3dAE}}

In this subsection, we prove Theorem \ref{thm:Thm3dAE}. Denote 
\[
\mathcal{\widetilde{N}}_{1}=\left\{ n\in\mathcal{N}_{3}:\,n=4^{a}n_{1},\,4\nmid n_{1}\Rightarrow n_{1}>n/\log^{2}n\right\} .
\]

\begin{lem}
\label{lem:N1DensityOne}The set $\mathcal{\widetilde{N}}_{1}$ is
a density one subset in $\mathcal{N}_{3}$. 
\end{lem}

\begin{proof}
Consider the complement set 
\[
\widetilde{\mathcal{N}}_{0}=\mathcal{N}_{3}\setminus\mathcal{\widetilde{N}}_{1}=\left\{ n\in\mathcal{N}_{3}:\,n=4^{a}n_{1},\,4\nmid n_{1}\Rightarrow n_{1}\le n/\log^{2}n\right\} .
\]
Let $X>0$, and assume that $n\in\widetilde{\mathcal{N}}_{0}$ and
$n\le X$. Then $a\le\log_{4}X$, and $n_{1}\le X/\log^{2}X.$ Since
$n$ is uniquely determined by $a$ and $n_{1}$, we conclude that
\[
\#\left\{ n\in\widetilde{\mathcal{N}}_{0}:\,n\le X\right\} \ll X/\log X
\]
so $\widetilde{\mathcal{N}}_{0}$ is a density zero set in $\mathcal{N}_{3}$,
and therefore $\mathcal{\widetilde{N}}_{1}$ is a density one subset
in $\mathcal{N}_{3}$. 
\end{proof}
\begin{cor}
The set
\[
\Lambda'=\left\{ \lambda\in\Lambda:\,n_{\lambda}\in\mathcal{\widetilde{N}}_{1}\right\} 
\]
is a density one subset in $\Lambda$.
\end{cor}

\begin{proof}
This is an immediate corollary of Lemma \ref{lem:N1DensityOne}, since
the elements of $\Lambda$ interlace with the elements of $\mathcal{N}_{3}$.
\end{proof}
\begin{lem}
\label{lem:3dAELemmaSmoothed}Let $\epsilon>0$. We have
\end{lem}

\[
\sup_{\substack{x\in\mathbb{T}^{2}\\
r>\lambda^{-1/6+\epsilon}
}
}\left|\frac{1}{\frac{4}{3}\pi r^{3}}\int_{\mathbb{T}^{3}}b_{x}^{\pm}\left(y\right)\left|g_{\lambda,L}\left(y\right)\right|^{2}\,dy-1\right|\to0
\]
as $\lambda\to\infty$ along $\Lambda'$, where $L=\lambda^{\delta}$,
$0<\delta<\epsilon/16$. 
\begin{proof}
Recall that by Lemma \ref{lem:ExpansionLemma}, for $T,r>0$ such
that $Tr>1$, we have
\begin{align}
 & \int_{\mathbb{T}^{3}}b_{x}^{\pm}\left(y\right)\left|g_{\lambda,L}\left(y\right)\right|^{2}\,dy-\frac{4}{3}\pi r^{3}\nonumber \\
 & =\frac{1}{\left\Vert G_{\lambda,L}\right\Vert ^{2}}\sum_{0<\left|\zeta\right|<T}\hat{b}_{x}^{\pm}\left(\zeta\right)e^{-i\left\langle x_{0},\zeta\right\rangle }\sum_{\begin{subarray}{c}
\left|\left|\xi\right|^{2}-\lambda\right|<L\\
\left|\left|\xi-\zeta\right|^{2}-\lambda\right|<L
\end{subarray}}\frac{1}{\left(\left|\xi\right|^{2}-\lambda\right)\left(\left|\xi-\zeta\right|^{2}-\lambda\right)}+O\left(r^{2}/T\right).\label{eq:AEexpansion}
\end{align}
Assume that $\lambda\in\Lambda'$, so that $n_{\lambda}\in\mathcal{\widetilde{N}}_{1}$,
and proceed as in the proof of Lemma \ref{lem:N1Lemma}: recall the
bound (\ref{eq:N1Bound}), and as before split the first summation
(\ref{eq:FirstSumN1}) on the right-hand-side of (\ref{eq:N1Bound})
into the sums $\sum^{4}+\sum^{5}$. The bound (\ref{eq:Sigma4Bound})
for $\sum^{4}$ still holds:
\[
\sum\nolimits ^{4}\frac{1}{\left(\left|\xi\right|^{2}-\lambda\right)^{2}}\ll_{\eta}\lambda^{2\delta+\eta}.
\]
In order to bound $\sum^{5}$, recall that by the bound (\ref{eq:Sigma5Bound})
we have
\[
\sum\nolimits ^{5}\frac{1}{\left(\left|\xi\right|^{2}-\lambda\right)^{2}}\ll_{\eta}\left\Vert G_{\lambda,L}\right\Vert ^{2}\frac{L\lambda^{\eta}}{r_{3}\left(n_{\lambda}\right)}.
\]
Write $n_{\lambda}=4^{a}n_{1}$ with $4\nmid n_{1}$ so that $n_{1}>n_{\lambda}/\log^{2}n_{\lambda}$.
Then by (\ref{eq:SiegelBound}), 
\[
\frac{L\lambda^{\eta}}{r_{3}\left(n_{\lambda}\right)}=\frac{\lambda^{\delta+\eta}}{r_{3}\left(n_{1}\right)}\ll\frac{\lambda^{\delta+\eta}}{n_{1}^{1/2-\eta}}\ll\frac{\lambda^{\delta+\eta}}{\left(n_{\lambda}/\log^{2}n_{\lambda}\right)^{1/2-\eta}}\ll\lambda^{\delta-\frac{1}{2}+2.5\eta}.
\]
The second summation on the right-hand-side of (\ref{eq:N1Bound})
satisfies (\ref{eq:SecondSumN1}), and therefore can be bounded similarly.
Hence,
\begin{equation}
\sum_{\begin{subarray}{c}
\left|\left|\xi\right|^{2}-\lambda\right|<L\\
\left|\left|\xi-\zeta\right|^{2}-\lambda\right|<L
\end{subarray}}\frac{1}{\left(\left|\xi\right|^{2}-\lambda\right)\left(\left|\xi-\zeta\right|^{2}-\lambda\right)}\ll_{\eta}\lambda^{2\delta+\eta}+\left\Vert G_{\lambda,L}\right\Vert ^{2}\lambda^{\delta-\frac{1}{2}+2.5\eta}.\label{eq:AEmainBound}
\end{equation}
Substituting the bound (\ref{eq:AEmainBound}) in (\ref{eq:AEexpansion}),
and using property (4) of Lemma \ref{lem:HarmanLemma} and the third
part of Lemma \ref{lem:LowerBndG3d}, we obtain
\begin{align*}
\int_{\mathbb{T}^{3}}b_{x}^{\pm}\left(y\right)\left|g_{\lambda,L}\left(y\right)\right|^{2}\,dy-\frac{4}{3}\pi r^{3} & \ll r^{3}T^{3}\left(\frac{\lambda^{2\delta+\eta}}{\left\Vert G_{\lambda,L}\right\Vert ^{2}}+\lambda^{\delta-\frac{1}{2}+2.5\eta}\right)+\frac{r^{2}}{T}\\
 & \ll r^{3}\left(T^{3}\lambda^{2\delta-\frac{1}{2}+2.5\eta}+\frac{1}{rT}\right).
\end{align*}
If $r>\lambda^{-1/6+\epsilon}$ , and $T=\lambda^{1/6-\frac{7}{8}\epsilon}>\frac{1}{r}$,
then taking $\eta=\epsilon$, we have
\[
\int_{\mathbb{T}^{3}}b_{x}^{\pm}\left(y\right)\left|g_{\lambda,L}\left(y\right)\right|^{2}\,dy-\frac{4}{3}\pi r^{3}\ll r^{3}\left(\lambda^{2\delta-\epsilon/8}+\lambda^{-\epsilon/8}\right).
\]
The statement of the lemma follows taking $0<\delta<\epsilon/16.$
\end{proof}
By a standard diagonal argument, and using property (1) of Lemma \ref{lem:HarmanLemma},
we deduce the following proposition.
\begin{prop}
\label{prop:MainProp3d-2}There exists a density one subset $\Lambda'\subseteq\Lambda$
such that for every $\epsilon>0$,
\[
\sup_{\substack{x\in\mathbb{T}^{2}\\
r>\lambda^{-1/6+\epsilon}
}
}\left|\frac{1}{\frac{4}{3}\pi r^{3}}\int_{B_{x}\left(r\right)}\left|g_{\lambda,L}\left(y\right)\right|^{2}\,dy-1\right|\to0
\]
as $\lambda\to\infty$ along $\Lambda'$, where $L=\lambda^{\delta}$,
$0<\delta<\epsilon/16$. 
\end{prop}

\begin{proof}[Proof of Theorem \ref{thm:Thm3dAE} ]
 By Lemma \ref{lem:ApproximationLemma3d}, for every $x\in\mathbb{T}^{3},\,r>\lambda^{-1/6+\epsilon},$
we have
\begin{align*}
\left|\frac{1}{\frac{4}{3}\pi r^{3}}\int_{B_{x}\left(r\right)}\left|g_{\lambda}\left(y\right)\right|^{2}\,dy-1\right| & \le\left|\frac{1}{\frac{4}{3}\pi r^{3}}\int_{B_{x}\left(r\right)}\left|g_{\lambda,L}\left(y\right)\right|^{2}\,dy-1\right|+O_{\eta}\left(\lambda^{\eta-\delta/2}\right)
\end{align*}
Hence,
\begin{align*}
\sup_{\substack{x\in\mathbb{T}^{2}\\
r>\lambda^{-1/6+\epsilon}
}
}\left|\frac{1}{\frac{4}{3}\pi r^{3}}\int_{B_{x}\left(r\right)}\left|g_{\lambda}\left(y\right)\right|^{2}\,dy-1\right| & \le\sup_{\substack{x\in\mathbb{T}^{2}\\
r>\lambda^{-1/6+\epsilon}
}
}\left|\frac{1}{\frac{4}{3}\pi r^{3}}\int_{B_{x}\left(r\right)}\left|g_{\lambda,L}\left(y\right)\right|^{2}\,dy-1\right|\\
 & +O_{\eta}\left(\lambda^{\eta-\delta/2}\right).
\end{align*}
Theorem \ref{thm:Thm3dAE} now follows from Proposition \ref{prop:MainProp3d-2},
choosing $\delta<\epsilon/16$ and $\eta<\delta/2$.
\end{proof}

\end{document}